\documentclass[a4paper,UKenglish,cleveref,autoref]{article}

\usepackage{graphicx}
\usepackage{amsmath}
\usepackage{amsthm}
\usepackage{hyperref}
\usepackage{algorithm}
\usepackage{authblk}
\usepackage[noend]{algpseudocode}

\newcommand{\RLZPRE}{\ensuremath{\text{RLZ}_{pref}}}
\newcommand{\OURS}{\ensuremath{\text{ReLZ}}}

\newtheorem{theorem}{Theorem}
\newtheorem{lemma}{Lemma}
\theoremstyle{definition}
\newtheorem{example}{Example}

\begin{document}

\title{Lempel--Ziv-like Parsing in Small Space\thanks{
D. Kosolobov supported by the Russian Science Foundation (RSF), project 18-71-00002 (for the upper bound analysis and a part of lower bound analysis). D. Valenzuela supported by the Academy of Finland (grant 309048). G. Navarro funded by Basal Funds FB0001 and Fondecyt Grant 1-200038, Chile. S.J. Puglisi supported by the Academy of Finland (grant 319454).}
}

%\titlerunning{Lempel--Ziv-like parsing in small space}        % if too long for running head

\author[1]{Dmitry Kosolobov}
\author[2]{Daniel Valenzuela}
\author[3]{Gonzalo Navarro}
\author[2]{Simon J. Puglisi}

\affil[1]{Ural Federal University, Ekaterinburg, Russia}
\affil[2]{Department of Computer Science, University of Helsinki, Finland}
\affil[3]{CeBiB, Department of Computer Science, University of Chile, Chile}

\date{}

\maketitle

\begin{abstract}
Lempel--Ziv (LZ77 or, briefly, LZ) is one of the most effective and widely-used
compressors for repetitive texts.
However, the existing efficient methods computing the exact LZ parsing have to
use linear or close to linear space to index the input text during the
construction of the parsing, which is prohibitive for long inputs.
An alternative is Relative Lempel--Ziv (RLZ), which indexes only a fixed reference
sequence, whose size can be controlled. Deriving the reference sequence by
sampling the text yields reasonable compression ratios for RLZ, but performance is
not always competitive with that of LZ and depends heavily on the similarity of
the reference to the text.
In this paper we introduce ReLZ, a technique that uses RLZ as a
preprocessor to approximate the LZ parsing using little memory. RLZ is first
used to produce a sequence of phrases, and these are regarded as metasymbols
that are input to LZ for a second-level parsing on a (most often) drastically
shorter sequence. This parsing is finally translated into one on the original
sequence.

We analyze the new scheme and prove that, like LZ, it achieves the $k$th order empirical
entropy compression $n H_k + o(n\log\sigma)$ with
$k = o(\log_\sigma n)$, where $n$ is the input length and $\sigma$ is the
alphabet size.
In fact, we prove this entropy bound not only for ReLZ but for a wide class of
LZ-like encodings.
Then, we establish a lower bound on ReLZ approximation ratio
showing that the number of phrases in it can be $\Omega(\log n)$ times
larger than the number of phrases in LZ.
Our experiments show that ReLZ is faster than
existing alternatives to compute the (exact or approximate) LZ parsing, at the
reasonable price of an approximation factor below $2.0$ in all tested scenarios,
and sometimes below $1.05$, to the size of LZ.

\end{abstract}

\section{Introduction}

The Lempel--Ziv (LZ77 or, shortly, LZ) parsing is a central algorithm in data compression: more
than 40 years since
its development \cite{LZ76,LZ77}, it is at the core of widely used compressors
(gzip, p7zip, zip, arj, rar...) and compressed formats (PNG, JPEG, TIFF,
PDF...), and receives much attention from researchers~\cite{BGI18,HPZ11,EMLZscan,PP18,LZENCY}
and developers in industry~\cite{BROTLI,ZSTANDARD}.

LZ parsing also has important theoretical properties. The number of phrases, $z$, into which LZ parses
a text has become the defacto measure of compressibility for dictionary-based methods~\cite{GNP18LATIN}, which in particular are most effective on highly repetitive sequences \cite{Nav12}. While there are
measures that are stronger than LZ \cite{SS82,KP18}, these are NP-complete to
compute. The LZ parse, which can be computed greedily in linear time \cite{KKP13}, is
then the stronger measure of dictionary-based compressibility on which to
build practical compressors.

Computing the LZ parsing requires the ability to find previous occurrences of
text substrings (their ``source''), so that the compressor can replace the
current occurrence (the ``target'') by a backward pointer to the source.
Parsing in linear time~\cite{KKP13} requires building data
structures that are proportional to the text size. When the text size exceeds
the available RAM, switching to external memory leads to prohibitive
computation times.
Compression utilities avoid this problem with different workarounds: by
limiting the sources to lie inside a short sliding window behind the
current text (see~\cite{GM07,Shields,Wyner}) or by partitioning the input into blocks
and compressing them independently. These variants can greatly degrade
compression performance,
however, and are unable in particular to exploit long-range repetitions.

Computation of LZ in compressed space was first
studied---to the best of our knowledge---in 2015:
A $(1+\epsilon)$-approximation scheme running in $O(n \log n)$
time\footnote{Hereafter, $\log$ denote logarithm with base 2 if it is not explicitly stated otherwise.}
with $O(z)$ memory, where $n$ is the length of the input, was proposed in~\cite{FGGK15},
and an exact algorithm with the same time $O(n\log n)$
and space bounded by the zeroth order empirical entropy was given in~\cite{PP15}.
The work~\cite{KK17} shows how to compute LZ-End---an LZ-like parsing---using
$O(z+\ell)$ compressed space and $O(n \log\ell)$ time w.h.p., where $\ell$ is the length
of the longest phrase. The recent studies on the Run-Length Burrows--Wheeler Transform
(RLBWT)~\cite{GNP18SODA} and its connections to LZ~\cite{BGI18,PP18} have enabled
the computation of the LZ parsing in compressed space $O(r)$ and time $O(n \log r)$
via RLBWT, where $r$ is the number of runs in the RLBWT.

Relative Lempel--Ziv (RLZ)~\cite{KPZ10} is a variant of LZ that exploits another
approach: it uses a fixed external sequence, the \emph{reference},
where the sources are to be found, which performs well when the reference is
carefully chosen~\cite{GPV16,LPMW16}. Different compressors have been
proposed based on this idea~\cite{DDN15,DG11,LPMW16}.
When random sampling of the text is used to build an artificial
reference the expected encoded size of the RLZ relates to the size of LZ~\cite{GPV16},
however the gap is still large in practice.
Some approaches have done a second pass of compression after RLZ~\cite{HPZ11,DDN15,FRESCO}
but they do not produce an LZ-like parsing that could be compared with LZ.

In this paper we propose \OURS{}, a parsing scheme that approximates the LZ parsing
by making use of RLZ as a preprocessing step. The phrases found by RLZ are treated
as metasymbols that form a new sequence, which is then parsed by LZ to discover
longer-range repetitions. The final result is then expressed as phrases of the original
text. The new sequence on which LZ is applied is expected to be much shorter than the
original, which avoids the memory problems of LZ. In exchange, the parsing we obtain is
limited to choose sources and targets formed by whole substrings found by RLZ, and is
therefore suboptimal.

We analyze the new scheme and prove that, like LZ, it achieves the $k$th order empirical
entropy compression $n H_k + o(n\log\sigma)$ (see definitions below) with
$k = o(\log_\sigma n)$, where $n$ is the length of the input string and $\sigma$ is the
alphabet size.
We show that it is crucial for this result to use the so-called rightmost
LZ encoding~\cite{AmirLandauUkkonen,BelazzouguiPuglisi,BCFG,bitoptimal,Larsson} in the
second step of \OURS{}; to our knowledge, this
is the first provable evidence of the impact of the rightmost encoding.
In fact, the result is more general: we show that the rightmost encoding of any LZ-like
parsing with $O(\frac{n}{\log_\sigma n})$ phrases achieves the entropy compression
when a variable length encoder is used for phrases. One might interpret this as
an indication of the weakness of the entropy measure. We then relate \OURS{} to LZ---the
de facto standard for dictionary-based compression---and prove that the
number of phrases in \OURS{} might be $\Omega(z\log n)$; we conjecture that this lower
bound is tight. The new scheme is tested and, in all the experiments, the number of
phrases found by \OURS{} never exceeded $2z$ (and it was around $1.05z$ in some cases).
In exchange, \OURS{} computes the parsing faster than the existing
alternatives.

The paper is organized as follows. In Sections~\ref{sec:prelim} and~\ref{sec:algorithm}
we introduce some notation and define the \OURS{} parsing and its variations.
Section~\ref{sec:upper-bound} contains the empirical entropy analysis.
Section~\ref{sec:lower-bound} establishes the $\Omega(z\log n)$ lower bound.
All experimental results are in Sections~\ref{sec:impl} and~\ref{sec:exp}.
%We conclude with some remarks in Section~\ref{sec:conclusion}.

\section{Preliminaries}\label{sec:prelim}

Let $T[1,n]$ be a string of length $n$ over the alphabet $\Sigma = \{1,2,\ldots,\sigma\}$;
$T[i]$ denotes the $i$th symbol of $T$ and $T[i,j]$ denotes the substring $T[i]T[i{+}1]\cdots T[j]$.
A substring $T[i,j]$ is \emph{a prefix} if $i=1$ and \emph{a suffix} if $j=n$.
The \emph{reverse} of $T$ is the string $T[n]T[n-1]\cdots T[1]$.
The concatenation of two strings $T$ and $T'$ is denoted by $T \cdot T'$ or simply $T T'$.

The \emph{zeroth order empirical entropy} (see~\cite{KosarajuManzini,Manzini}) of $T[1,n]$ is defined as $H_0(T) = \sum_{c \in \Sigma} \frac{n_c}{n} \log \frac{n}{n_c}$, where $n_c$ is the number of symbols $c$ in $T$ and $\frac{n_c}{n}\log\frac{n}{n_c} = 0$ whenever $n_c = 0$. For a string $W$, let $T_W$ be a string formed by concatenating all symbols immediately following occurrences of $W$ in $T[1,n]$; e.g., $T_{ab} = aac$ for $T = abababc$. The \emph{$k$th order empirical entropy} of $T[1,n]$ is defined as $H_k(T) = \sum_{W \in \Sigma^k} \frac{|T_W|}{n} H_0(T_W)$, where $\Sigma^k$ is the set of all strings of length $k$ over $\Sigma$ (see~\cite{KosarajuManzini,MakinenNavarro,Manzini}). If $T$ is clear from the context, $H_k(T)$ is denoted by $H_k$. It is well known that $\log\sigma \ge H_0 \ge H_1 \ge \cdots$ and $H_k$ makes sense as a measure of string compression only for $k < \log_\sigma n$ (see~\cite{Gagie} for a deep discussion).

The \emph{LZ parsing}~\cite{LZ76} of $T[1,n]$ is a sequence of non-empty \emph{phrases}
(substrings) $\mathit{LZ}(T) = (P_1,P_2,\ldots,P_z)$ such that $T = P_1 P_2 \cdots P_z$, built
as follows.
Assuming we have already parsed $T[1,i-1]$, producing $P_1, P_2,\ldots,P_{j-1}$, then $P_j$ is set to the longest prefix of $T[i,n]$ that has a previous occurrence in $T$ that starts before position $i$.
Such a phrase $P_j$ is called a \emph{copying phrase}, and its previous occurrence in $T$ is called \emph{the source} of $P_j$.
When the longest prefix is of length zero, the next phrase is the single symbol $P_j = T[i]$, and $P_j$ is called a \emph{literal phrase}.
This greedy parsing strategy yields the least number of phrases (see \cite[Th. 1]{LZ76}).

LZ compression consists in replacing copying phrases by backward pointers to
their sources in $T$, and $T$ can obviously be reconstructed in linear time from
these pointers. A natural way to encode the phrases is as pairs of integers:
for copying phrases $P_j$, a pair $(d_j,\ell_j)$ gives the distance to the source
and its length, i.e., $\ell_j = |P_j|$ and $T[|P_1 \cdots P_{j-1}| - d_j + 1, n]$ is prefixed by $P_j$;
for literal phrases $P_j = c$, a pair $(c, 0)$ encodes the symbol $c$ as an integer.
Such encoding is called \emph{rightmost} if the numbers $d_j$ in
all the pairs $(d_j, \ell_j)$ are minimized, i.e., the rightmost sources are chosen.

When measuring the compression efficiency of encodings, it is natural to assume that
$\sigma$ is a non-decreasing function of $n$.
In such premises, if each $d_j$ component occupies $\lceil\log n\rceil$ bits and
each $\ell_j$ component takes $O(1 + \log\ell_j)$ bits, then it is known
that the size of the LZ encoding is upperbounded by $nH_k + o(n\log\sigma)$ bits,
provided $k$ is a function of $n$ such that
$k = o(\log_\sigma n)$; see \cite{Ganczorz,KosarajuManzini,OchoaNavarro}.
In the sequel we also utilize a slightly different encoding that, for each $d_j$, uses
a universal code \cite{Elias,Levenshtein} taking $\log d_j + O(1 + \log\log d_j)$ bits.

Other parsing strategies that do not necessarily choose the longest prefix of
$T[i,n]$ are valid, in the sense that $T$ can be recovered from the
backward pointers. Those are called \emph{LZ-like
parses}. Some examples are LZ-End~\cite{LZEND}, which forces sources to finish
at the end of a previous phrase, LZ77 with sliding window~\cite{LZ77}, which restricts the
sources to start in $T[i-w,i-1]$ for a fixed windows size $w$, and the bit-optimal
LZ~\cite{bitoptimal,Kosolobov}, where the phrases are chosen to minimize the encoding size
for a given encoder of pairs.

The RLZ parsing~\cite{KPZ10} of $T[1,n]$ with reference $R[1,\ell]$
is a sequence of phrases $\mathit{RLZ}(T,R) = (P_1,P_2,\ldots,P_z)$ such that
$T = P_1 P_2 \cdots P_z$, built as follows:
Assuming we have already parsed $T[1,i-1]$, producing $P_1,P_2,\ldots,P_{j-1}$, then $P_j$
is set to the longest prefix of $T[i,n]$ that is a substring of $R[1,\ell]$;
by analogy to the LZ parsing, $P_j$ is a \emph{copying phrase} unless it is of length zero; in the latter case we set
$P_j = T[i]$, a \emph{literal phrase}.
%\footnote{If all the distinct symbols in $T$ are present in $R$, then all the phrases are copying phrases.}
Note that RLZ does not produce an LZ-like parsing as we have defined it.

\section{\OURS{} Parsing}\label{sec:algorithm}
\label{sec:ReLZ}
First we present \RLZPRE{} ~\cite{V16}, a variant of RLZ that instead of
using an external reference uses a prefix
of the text as a reference to produce an LZ-like parsing.
The \RLZPRE{} parsing of $T$, given a parameter $\ell$, is defined as
$RLZ_{prefix}(T,\ell) = LZ(T[1,\ell]) \cdot RLZ(T[\ell+1,n],T[1,\ell])$.
That is, we first compress $T[1,\ell]$ with LZ, and then use that prefix as
the reference to compress the rest, $T[\ell+1,n]$, with RLZ. Note that
\RLZPRE{} is an LZ-like parsing.

The \OURS{} algorithm works as follows. Given a text $T[1,n]$ and a prefix size
$\ell$, we first compute the \RLZPRE{} parsing $(P_1, P_2, \ldots, P_{z'})$
(so that $T = P_1 P_2 \cdots P_{z'}$).
Now we consider the phrases $P_j$ as atomic metasymbols, and define a string
$T'[1,z']$ such that, for every $i$ and $j$, $T'[i] = T'[j]$ iff $P_i = P_j$.
Then we compress $T'[1,z']$ using LZ, which yields a parsing
$(P'_1, P'_2, \ldots, P'_{\hat{z}})$ of $T'$.
Finally, the result is transformed into an LZ-like parsing of $T$ in a
straightforward way: each literal phrase $P'_j$ corresponds to a single phrase $P_i$
and, thus, is left unchanged; each copying phrase $P'_j$ has a source $T'[p,q]$ and
is transformed accordingly into a copying phrase in $T$ with the source $T[p',q']$, where
$p' = |P_1 P_2\cdots P_{p-1}| + 1$ and $q' = |P_1 P_2\cdots P_q|$.
%Algorithm~\ref{algorithm1} shows the pseudo-code and
Figure~\ref{fig:example} shows an example.

\begin{figure}[ht]
	\centering
	\includegraphics[width=\textwidth]{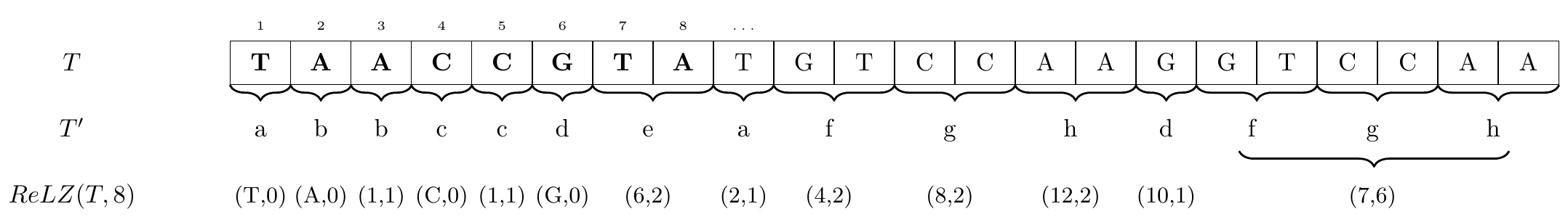}
	\caption{
    An example of \OURS{}, using prefix size $\ell=8$.
    The first line below the text shows the string $T'$ corresponding to the \RLZPRE{} parsing. Note that the substring
     ``GTCCAA'' occurs twice, but \RLZPRE{} misses this repetition because there is no similar substring in the reference.
    Nonetheless, both occurrences are parsed identically.
    The string $T'$ is then parsed using LZ.
    The latter captures the repetition of the sequence ``fgh'', and when this parsing is remapped to the
    original text, it captures the repetition of ``GTCCAA''.
}
	\label{fig:example}
\end{figure}

Since both LZ~\cite{KKP13} and RLZ~\cite{KPZ11}
run in linear time, \OURS{} can also be implemented in time $O(n)$.

Obviously, the first $\ell$ symbols do not necessarily make a good reference
for the RLZ step in \OURS{}. In view of this, it seems practically relevant to
define the following variant of \OURS{}:
given a parameter $\ell = o(n)$, we first sample in a certain way (for instance,
randomly as in~\cite{GPV16}) disjoint substrings of $T$ with total length $\ell$,
then concatenate them making a string $A$ of length $\ell$, and apply \OURS{} to the
string $A T$; the output encoding of $T$ is the $\lceil\log n\rceil$-bit number
$\ell$ followed by an encoding of the LZ-like parsing of $A T$ produced by \OURS{}.
Nevertheless, throughout the paper we concentrate only on the first version of \OURS,
which generates an LZ-like parsing. This choice is justified by two observations: first,
it is straightforward that the key part in any analysis of the second \OURS{} variant
is in the analysis of \OURS{} for the string $A T$; and second, our experiments
on real data comparing known sampling methods (see Section~\ref{subsec:ref-constr}) show that the first
version of \OURS{} leads to better compression, presumably because the improvements
made by the sampling in the RLZ step do not compensate for the need to keep the
reference $A$.

\section{Empirical Entropy Upper Bound}\label{sec:upper-bound}

Our entropy analysis relies on the following lemmas by Gagie~\cite{Gagie}, Ochoa and Navarro~\cite{OchoaNavarro}, and Ga{\'n}czorz~\cite{Ganczorz}.

\begin{lemma}[{\cite[Th.~1]{Gagie}}]\label{lem:entropy-markov}
For any string $T[1,n]$ and any integer $k \ge 0$, $n H_k(T) = \min\{\log(1/Pr(Q\text{ emits }T))\}$, where the minimum is over all $k$th order Markov processes $Q$.
\end{lemma}

\begin{lemma}[{\cite{Ganczorz} and \cite[Lm.~3]{OchoaNavarro}}]\label{lem:parsing-entropy}
Let $Q$ be a $k$th order Markov process. Any parsing $T = P_1 P_2 \cdots P_c$ of a given string $T[1,n]$ over the alphabet $\{1,2,\ldots,\sigma\}$, where all $P_i$ are non-empty, satisfies:
$$
\sum_{i=1}^c \log\frac{c}{c_i} \le \log\frac{1}{Pr(Q\text{ emits }T)} + O(ck\log\sigma + c\log\frac{n}{c}),
$$
where $c_i$ is the number of times $P_i$ occurs in the sequence $P_1, P_2, \ldots, P_c$.
\end{lemma}

Recall that in this discussion $\sigma$ and $k$ in $H_k$ both are functions of $n$. Now we are to prove that, as it turns out, the $k$th order empirical entropy is easily achievable by any LZ-like parsing in which the number of phrases is $O(\frac{n}{\log_\sigma n})$: it suffices to use the rightmost encoding and to spend at most $\log d_j + O(1 + \log\log d_j + \log\ell_j)$ bits for every pair $(d_j,\ell_j)$ corresponding to a copying phrase (for instance, applying for $d_j$ and $\ell_j$ universal codes, like Elias's~\cite{Elias} or Levenshtein's~\cite{Levenshtein}). In the sequel we show that, contrary to the case of LZ (see \cite{KosarajuManzini,OchoaNavarro}), it is not possible to weaken the assumptions in this result---even for \OURS{}---neither by using a non-rightmost encoding nor by using $\log n + O(1 + \log\ell_j)$ bits for the pairs $(d_j, \ell_j)$.

\begin{lemma}\label{lem:entropy}
Fix a constant $\alpha > 0$. Given a string $T[1,n]$ over the alphabet $\{1,2,\ldots,\sigma\}$ with $\sigma \le O(n)$ and its LZ-like parsing $T = P_1 P_2 \cdots P_c$ such that $c \le \frac{\alpha n}{\log_\sigma n}$, the rightmost encoding of the parsing in which every pair $(d_j,\ell_j)$ corresponding to a copying phrase takes $\log d_j + O(1 + \log\log d_j + \log\ell_j)$ bits occupies at most $nH_k + o(n\log\sigma)$ bits, for $k = o(\log_\sigma n)$.
\end{lemma}
\begin{proof}
First, let us assume that $k$ is a positive function of $n$, $k > 0$. Since $k = o(\log_\sigma n)$, it implies that $\log_\sigma n = \omega(1)$ and $\sigma = o(n)$. Therefore, all literal phrases occupy $O(\sigma\log\sigma) = o(n\log\sigma)$ bits. For $i \in \{1,2,\ldots,c\}$, denote by $c_i$ the number of times $P_i$ occurs in the sequence $P_1,P_2, \ldots, P_c$. Let $P_{i_1}, P_{i_2}, \ldots, P_{i_{c_i}}$ be the subsequence of all phrases equal to $P_i$. Since the encoding we consider is rightmost, we have $d_{i_1} + d_{i_2} + \cdots + d_{i_{c_i}} \le n$. Therefore, by the concavity of the function $\log$, we obtain $\log d_{i_1} + \log d_{i_2} + \cdots + \log d_{i_{c_i}} \le c_i \log\frac{n}{c_i} = c_i \log\frac{n}{c} + c_i \log\frac{c}{c_i}$. Similarly, we deduce $\log\log d_{i_1} + \log\log d_{i_2} + \cdots + \log\log d_{i_{c_i}} \le c_i \log\log\frac{n}{c_i}$ and $\sum_{j=1}^c \log\ell_j \le c\log\frac{n}{c}$. Hence, the whole encoding occupies $\sum_{i=1}^c (\log\frac{c}{c_i} + O(\log\log\frac{n}{c_i})) + O(c \log\frac{n}{c}) + o(n\log\sigma)$ bits. By Lemmas~\ref{lem:entropy-markov} and~\ref{lem:parsing-entropy}, this sum is upperbounded by
\begin{equation}\label{eq:entropy-estimate}
nH_k + O(c\log\frac{n}{c} + ck\log\sigma + \sum_{i=1}^c \log\log\frac{n}{c_i}) + o(n\log\sigma).
\end{equation}
It remains to prove that all the terms under the big-O are $o(n\log\sigma)$. Since $k = o(\log_\sigma n)$ and $c \le \frac{\alpha n}{\log_\sigma n}$, we have $ck\log\sigma \le o(n\log\sigma)$. As $c\log\frac{n}{c}$ is an increasing function of $c$ when $c < n/2$, we obtain $c\log\frac{n}{c} \le O(\frac{n}{\log_\sigma n}\log\log_\sigma n) = o(n\log\sigma)$. Further, $\log\log\frac{n}{c_i} = \log(\log\frac{n}{c} + \log\frac{c}{c_i}) \le \log\log\frac{n}{c} + O(\log\frac{c}{c_i} / \log\frac{n}{c})$ due to the inequality $\log(x + d) \le \log x + \frac{d\log e}{x}$. The sum $\sum_{i=1}^c \log\log\frac{n}{c}$ is upperbounded by $c\log\frac{n}{c} = o(n\log\sigma)$. The sum $\sum_{i=1}^c \log\frac{c}{c_i} / \log\frac{n}{c}$ is upperbounded by $(c\log c) / \log\frac{n}{c}$, which can be estimated as $O((n\log\sigma) / \log\log_\sigma n)$ because $c \le \frac{\alpha n}{\log_\sigma n}$. Since $\log_\sigma n = \omega(1)$, this is again $o(n\log\sigma)$.

Now assume that $k = 0$; note that in this case $k = o(\log_\sigma n)$ even for $\sigma = \Omega(n)$. It is sufficient to consider only the case $\sigma > 2^{\sqrt{\log n}}$ since, for $\sigma \le 2^{\sqrt{\log n}}$, we have $\sigma\log\sigma = o(n\log\sigma)$ and $\log_\sigma n \ge \sqrt{\log n} = \omega(1)$ and, hence, the above analysis is applicable. As $\sigma$ can be close to $\Theta(n)$, the literal phrases might now take non-negligible space. Let $A$ be the subset of all symbols $\{1,2,\ldots,\sigma\}$ that occur in $T$. For $a \in A$, denote by $i_a$ the leftmost phrase $P_{i_a} = a$. Denote $C = \{1,2,\ldots,c\} \setminus \{i_a \colon a \in A\}$, the indices of all copying phrases. The whole encoding occupies at most $\sum_{i \in C} (\log d_i + O(1 + \log\log d_i + \log\ell_i)) + |A|\log\sigma + O(|A|(1 + \log\log\sigma))$ bits, which is upperbounded by $\sum_{i \in C} \log d_i + |A|\log|A| + O(n + n\log\log n + c\log\frac{n}{c}) + |A|\log\frac{\sigma}{|A|}$. Observe that $n\log\log n \le o(n\sqrt{\log n}) \le o(n\log\sigma)$. Further, we have $c\log\frac{n}{c} \le n$ and $|A|\log\frac{\sigma}{|A|} \le \sigma$, both of which are $o(n\log\sigma)$ since $\sigma \le O(n) \le o(n\log\sigma)$. It remains to bound $\sum_{i \in C} \log d_i + |A|\log|A|$ with $nH_0 + o(n\log\sigma)$.

Let us show that $|A|\log|A| \le \sum_{a \in A} \log\frac{n}{c_{i_a}}$. Indeed, we have $\sum_{a \in A} \log\frac{n}{c_{i_a}} = \log(n^{|A|} / \prod_{a \in A} c_{i_a})$, which, since $\sum_{a\in A} c_{i_a} \le n$, is minimized when all $c_{i_a}$ are equal to $\frac{n}{|A|}$ so that $\sum_{a \in A} \log\frac{n}{c_{i_a}} \ge |A|\log|A|$. For $i \in C$, denote by $c'_i$ the number of copying phrases equal to $P_i$, i.e., $c'_i = c_i$ if $|P_i| > 1$, and $c'_i = c_i - 1$ otherwise (note that $c'_i > 0$ for all $i \in C$). As in the analysis for $k > 0$, we obtain $\sum_{i \in C} \log d_i \le \sum_{i \in C} \log\frac{n}{c'_i}$. Fix $i \in C$ such that $|P_i| = 1$. Using the inequality $\log(x - d) \ge \log x - \frac{d\log e}{x - d}$, we deduce $\log\frac{n}{c'_i} = -\log(\frac{c_i}{n} - \frac{1}{n}) \le \log\frac{n}{c_i} + \frac{\log e}{c'_i}$. Therefore, $|A|\log|A| + \sum_{i \in C} \log\frac{n}{c'_i} \le \sum_{i=1}^c \log\frac{n}{c_i} + c\log e = \sum_{i=1}^c \log\frac{c}{c_i} + c\log\frac{n}{c} + O(n)$. By Lemmas~\ref{lem:entropy-markov} and~\ref{lem:parsing-entropy}, this is upperbounded by \eqref{eq:entropy-estimate}. As $k = 0$, the terms under the big-O of \eqref{eq:entropy-estimate} degenerate to $c\log\frac{n}{c} + \sum_{i=1}^c \log\log\frac{n}{c_i}$, which is $O(n\log\log n) \le o(n\log\sigma)$.
%\qed
\end{proof}

It follows from the proof of Lemma~\ref{lem:entropy} that, instead of the strict rightmost encoding, it is enough to choose, for each copying phrase $P_j$ of the LZ-like parsing, the closest preceding equal phrase---i.e., $P_i = P_j$ with maximal $i < j$---as a source of $P_j$, or any source if there is no such $P_i$. This observation greatly simplifies the construction of an encoding that achieves the $H_k$ bound. Now let us return to the discussion of the \OURS{} parsing.

\begin{lemma}\label{lem:relz-phrases}
The number of phrases in the \OURS{} parsing of any string $T[1,n]$ over the alphabet $\{1,2,\ldots,\sigma\}$ is at most $\frac{9 n}{\log_\sigma n}$, independent of the choice of the prefix parameter $\ell$.
\end{lemma}
\begin{proof}
For $\sigma \ge n^{1/9}$, we have $\frac{9n}{\log_\sigma n} \ge n$ and, hence, the claim is obviously true. Assume that $\sigma < n^{1/9}$. As $\sigma \ge 2$, this implies $n > 2^9 = 512$. Suppose that $T = P_1 P_2 \cdots P_{\hat{z}}$ is the \OURS{} parsing, for a given prefix size $\ell$. We are to prove that there are at most $1 + 2\sqrt{n}$ indices $j < \hat{z}$ such that $|P_j| < \frac{1}{4}\log_\sigma n$ and $|P_{j+1}| < \frac{1}{4} \log_\sigma n$. This will imply that every phrase of length less than $\frac{1}{4} \log_\sigma n$ is followed by a phrase of length at least $\frac{1}{4} \log_\sigma n$, except for at most $2 + 2\sqrt{n}$ exceptions ($1 + 2\sqrt{n}$ plus the last phrase). Therefore, the total number of phrases is at most $2 + 2\sqrt{n} + \frac{2n}{(1/4)\log_\sigma n} = 2 + 2\sqrt{n} + \frac{8n}{\log_\sigma n}$; the term $2 + 2\sqrt{n}$ is upperbounded by $\frac{n}{\log_\sigma n}$ since $n > 512$, and thus, the total number of phrases is at most $\frac{9n}{\log_\sigma n}$ as required.

It remains to prove that there are at most $1 + 2\sqrt{n}$ pairs of ``short'' phrases $P_j, P_{j+1}$. First, observe that any two equal phrases of the LZ parsing of the prefix $T[1,\ell]$ are followed by distinct symbols, except, possibly, for the last phrase. Hence, there are at most $1 + \sum_{k=1}^{\lfloor\frac{1}{4}\log_\sigma n\rfloor} \sigma^{k+1} \le 1 + \sigma^2\sigma^{\frac{1}{4}\log_\sigma n} \le 1 + n^{2/9}n^{1/4} \le 1 + \sqrt{n}$ phrases of length less than $\frac{1}{4}\log_\sigma n$ in the LZ parsing of $T[1,\ell]$. Further, there cannot be two distinct indices $j < j' < \hat{z}$ such that $P_j = P_{j'}$, $P_{j+1} = P_{j'+1}$, and $|P_1 P_2 \cdots P_{j-1}| \ge \ell$ (i.e., $P_j$ and $P_{j'}$ both are inside the $T[\ell + 1, n]$ part of $T$): indeed, the RLZ step of \OURS{} necessarily parses the substrings $P_j$, $P_{j+1}$ and $P_{j'}$, $P_{j'+1}$ equally, and then, the LZ step of \OURS{} should have realized during the parsing of $P_{j'} P_{j'+1}$ that this string occurred previously in $P_j P_{j+1}$ and it should have generated a new phrase comprising $P_{j'} P_{j'+1}$. Therefore, there are at most $\sigma^{\frac{1}{4}\log_\sigma n} \sigma^{\frac{1}{4}\log_\sigma n} = \sqrt{n}$ indices $j < \hat{z}$ such that $P_j$ and $P_{j+1}$ both are ``short'' and $P_jP_{j+1}$ is inside $T[\ell + 1, n]$. In total, we have at most $1 + 2\sqrt{n}$ phrases $P_j$ such that $|P_j| < \frac{1}{4}\log_\sigma n$ and $|P_{j+1}| < \frac{1}{4} \log_\sigma n$.
%\qed
\end{proof}

Lemmas~\ref{lem:entropy} and~\ref{lem:relz-phrases} immediately imply the following theorem.

\begin{theorem}\label{thm:relz-entropy}
Given a string $T[1,n]$ over the alphabet $\{1,2,\ldots,\sigma\}$ with $\sigma \le O(n)$, the rightmost encoding of any \OURS{} parsing of $T$ in which every pair $(d_j,\ell_j)$ corresponding to a copying phrase takes $\log d_j + O(1 + \log\log d_j + \log\ell_j)$ bits occupies $nH_k + o(n\log\sigma)$ bits, for $k = o(\log_\sigma n)$.
\end{theorem}

For LZ, it is not necessary to use neither the rightmost encoding nor less than $\log n$ bits for the $d_j$ components of pairs in order to achieve the $k$th order empirical entropy with $k = o(\log_\sigma n)$. In view of this, the natural question is whether the \OURS{} really requires these two assumptions of Theorem~\ref{thm:relz-entropy}. The following example shows that indeed the assumptions cannot be simply removed.

\begin{example}
Fix an integer $b \ge 3$. Our example is a string of length $n = b2^{2b} + 2^b$ over the alphabet $\{0,1,2\}$. Denote by $a_1, a_2, \ldots, a_{2^b}$ all possible binary strings of length $b$. Put $A = a_1 2 a_2 2 \cdots a_{2^b} 2$ ($a_1, a_2, \ldots, a_{2^b}$ separated by $2$s). The example string is $T = A B_1 B_2 \cdots B_{2^b-1}$, where each string $B_h$ is the concatenation of $a_1, a_2, \ldots, a_{2^b}$ in a certain order such that every pair of distinct strings $a_i$ and $a_j$ can be concatenated in $B_1 B_2 \cdots B_{2^b-1}$ at most once. More precisely, we have $2^b-1$ permutations $\pi_h$ of the set $\{1,2,\ldots,2^b\}$, for $1 \le h < 2^b$, such that $B_h = a_{\pi_h(1)} a_{\pi_h(2)} \cdots a_{\pi_h(2^b)}$ and, for every integers $i$ and $j$ with $1 \le i < j \le 2^b$, at most one $h$ satisfies $\pi_h(2^b) = i$ and $\pi_{h+1}(1) = j$, or $\pi_h(k) = i$ and $\pi_h(k+1) = j$, for some $k < 2^b$.

Let us show that the permutations $\pi_h$ can be constructed from a decomposition of the complete directed graph $K^*_{2^b}$ with $2^b$ vertices into $2^b - 1$ disjoint Hamiltonian directed cycles; Tillson~\cite{Tillson} proved that such decomposition always exists for $2^b \ge 8$. (Note that the number of edges in $K^*_{2^b}$ is $2^{2b} - 2^b$ and every Hamiltonian cycle contains $2^b$ edges, so $2^b - 1$ is the maximal number of disjoint cycles.) Denote the vertices of $K^*_{2^b}$ by $1,2,\ldots,2^b$. Every Hamiltonian cycle naturally induces $2^b$ permutations $\pi$: we arbitrarily choose $\pi(1)$ and then, for $k > 1$, put $\pi(k)$ equal to the vertex number following $\pi(k-1)$ in the cycle. Since the cycles are disjoint, any two distinct numbers $i$ and $j$ cannot occur in this order in two permutations corresponding to different cycles, i.e., $\pi_h(k) = i$ and $\pi_h(k+1) = j$, for some $k$, can happen at most in one $h$; further, we put $\pi_1(1) = 1$ and, for $h > 1$, we assign to $\pi_h(1)$ the vertex number following $\pi_{h-1}(2^b)$ in the cycle corresponding to $\pi_{h-1}$, so that $\pi_{h-1}(2^b) = i$ and $\pi_h(1) = j$, for fixed $i$ and $j$, can happen in at most one $h$.

Put $\ell = |A|$, the parameter of \OURS{}. Clearly, the RLZ step of \OURS{} parses $B_1 B_2 \cdots B_{2^b-1}$ into $2^b(2^b - 1)$ phrases of length $b$. By construction, all equal phrases in the parsing are followed by distinct phrases. Therefore, the LZ step of \OURS{} does not reduce the number of phrases. Suppose that the source of every copying phrase is in $A$ (so, we assume that the encoding is not rightmost) and we spend at least $\log d_j$ bits to encode each pair $(d_j,\ell_j)$ corresponding to a copying phrase. Therefore, the encoding overall occupies at least $\sum_{i=1}^{2^b(2^b-1)} \log(ib)$ bits, which can be lowerbounded by $\sum_{i=1}^{2^{2b}-2^b} \log i = \log((2^{2b} - 2^b)!) = (2^{2b} - 2^b)\log(2^{2b} - 2^b) - O(2^{2b})$. Recall that $n = b2^{2b} + 2^b$ and, hence, $b = \Theta(\log n)$, $2^{2b} = o(n)$, and $2^b\log(2^{2b} - 2b) = o(n)$. Thus, $(2^{2b} - 2^b)\log(2^{2b} - 2^b) - O(2^{2b}) \ge 2^{2b}\log(2^{2b} - 2^b) - o(n)$. By the inequality $\log(x - d) \ge \log x - \frac{d\log e}{x - d}$, the latter is lowerbounded by $2^{2b}\log(2^{2b}) - O(2^{2b} 2^b / (2^{2b} - 2^b)) - o(n) = 2b 2^{2b} - o(n) = 2n - o(n)$. On the other hand, we obviously have $H_0(T) \approx 1$ and, thus, $nH_0(T) = n - o(n)$. Therefore, the non-rightmost encoding, which forced us to use at least ${\sim}\log n$ bits for many pairs $(d_j,\ell_j)$, does not achieve the zeroth empirical entropy of $T$.
\end{example}

\section{Lower Bound}\label{sec:lower-bound}

We have not been able to upper bound the number of phrases $\hat{z}$ resulting
from \OURS{} in terms of the optimal number $z$ of phrases produced by the LZ
parsing of $T$. Note that, in the extreme cases $\ell=n$ and $\ell=0$, we have
$\hat{z}=z$, but these are not useful choices: in the former case we apply $LZ(T)$ in
the first phase and in the latter case we apply $LZ(T')$, with $T' \approx T$, in
the second phase. In this section, we obtain the following lower bound.

\begin{theorem}
There is an infinite family of strings over the alphabet $\{0,1,2\}$ such that,
for each family string $T[1,n]$, the number of phrases in its \OURS{} parse
(with an appropriate parameter $\ell = o(n)$) and its LZ parse---respectively,
$\hat{z}$ and $z$---are related as $\hat{z}=\Omega (z \log n)$.
\end{theorem}
\begin{proof}
The family contains, for each even positive integer $b$, a string $T$ of length
$\Theta(b^2 2^b)$ built as follows.
Let $A$ be the concatenation of all length-$b$ binary strings in the lexicographic
order, separated by the special symbol $2$ and with $2$ in the end.
Let $S$ be the concatenation of all length-$b$ binary strings in the
lexicographic order. (E.g., $A = 002012102112$ and $S = 00011011$ for $b = 2$.)
Finally, let $S_i$ be $S$ cyclically shifted to the left $i$
times, i.e., $S_i = S[i+1,|S|]\cdot S[1,i]$.
Then, put $T = A S_1 S_2 \cdots S_{\frac{b}{2}}$
and we use $\ell=|A|$ as a parameter for \OURS{}. So $n=|T|=\Theta(b^2 2^b)$ and $\log n = \Theta(b)$.
% and hence $\ell = |A| = \Theta(b 2^b) = \Theta(n/\log n)$.
We are to prove that
$z = |LZ(T)| = O(2^b)$ and $\hat{z} = |\mathit{\OURS}(T,\ell)| = \Omega (b 2^b)$, which
will imply $\hat{z} = \Omega(z\log n)$, thus concluding the proof.

By \cite[Th. 1]{LZ76}, the LZ parse has the smallest
number of phrases among all LZ-like parses of $T$. Therefore, to show that
$z = O(2^b)$, it suffices to describe an LZ-like parse of $T$ with $O(2^b)$ phrases.
Indeed, the prefix $A$ can be parsed into $O(2^b)$ phrases as follows:
all symbols $2$ form phrases of length one;
the first length-$b$ substring $00\cdots 0$ can be parsed into $b$ literal phrases $0$;
every subsequent binary length-$b$ substring $a_1 a_2 \cdots a_b$ with $a_k = 1$
and $a_{k+1} = a_{k+2} = \cdots = a_b = 0$,
for some $k \in \{1,2,\ldots,b\}$, can be parsed into the copying phrase $a_1 a_2\cdots a_{k-1}$
(which must be a prefix of the previous length-$b$ binary substring $a_1 a_2 \cdots a_{k-1}011\cdots 1$,
due to the lexicographic order in $A$), the literal phrase $1$, and the copying phrase
$a_{k+1}a_{k+2}\cdots a_b = 00\cdots 0$. The string $S_1$ can be analogously parsed into
$O(2^b)$ phrases. Each $S_i$, $i > 2$, can be expressed as two phrases
that point to $S_1$. Thus, we obtain $z \leq |\mathit{LZ}(A)| + |\mathit{LZ}(S_1)| + 2(b/2 - 1) = O(2^b)$.

Now consider $\hat{z}$.
The first phase of  $\mathit{\OURS}(T,\ell)$ parses $T$ into phrases whose sources are restricted
to be within $T[1,\ell]=A$. Therefore, it is clear that, for any
$i \in \{1,2,\ldots,\frac{b}{2}\}$, $S_i$ will be parsed into $2^b$ strings
of length $b$, because every length-$b$ string is in $A$ separated by 2s.
In what follows we show that the second phase of \OURS{} cannot further reduce the
number of phrases and, hence, $\hat{z} \ge \frac{b}2 2^b = \Omega(b 2^b)$ as required.

Let us consider $S_i$ and $S_j$, for some $i < j$,
and let us denote their parsings by $R_1, R_2, \ldots, R_{2^b}$
and $R'_1, R'_2, \ldots, R'_{2^b}$, respectively.
Suppose that there are indices $k$ and $h$ such that $R_k = R'_h$.
We are to prove that $R_{k+1} \neq R'_{h+1}$ (assuming
$R_{k+1}$ is the length-$b$ prefix of $S_{i+1}$ if $k = 2^b$, and analogously for $h = 2^b$).
This will imply that all phrases produced
by the second phase of \OURS{} on the string of metasymbols are of length one.

Consider the case $k < 2^b$ and $h < 2^b$.
Let us interpret the bitstrings of length $b$ as numbers so that the most
and the least significant bits are indexed by $1$ and $b$, respectively;%
\footnote{To conform with the indexation scheme used throughout the paper, we
do not follow the standard practice to index the least significant bit as zeroth.}
e.g., in the string $01$, for $b = 2$, the least significant bit is the second
symbol and equals $1$.
In this way we can see $S = Q_1 Q_2 \cdots Q_{2^b}$, where $|Q_1| = \cdots = |Q_{2^b}| = b$, as generated
by adding 1 to the previous bitstring, starting from $Q_1 = 00\cdots0$.
Now, the $(b-i)$th symbols of $R_k$ and $R_{k+1}$ are different since they correspond to the lowest bit
in $Q_1,Q_2,\ldots,Q_{2^b}$ (thus, the $(b-i)$th symbol alternates in $R_1,\ldots,R_{2^b}$, starting from $0$).
Suppose that the $(b-i)$th symbols of $R'_h$ and $R'_{h+1}$ also differ (otherwise our claim is trivially true).
Since $0 < i < j$, this implies that the symbols $b, b-1, \ldots, b-i+1$ in $R'_h$ and
$1, 2, \ldots,  b-j$ in $R'_{h+1}$ all are equal to $1$ (this cascade of ones triggers the change in
the $(b-i)$th symbol of $R'_{h+1}$), the symbols $b, b-1, \ldots, b-i+1$ in $R'_{h+1}$ equal $0$
(as a result of the ``collapse'' of the cascade), and the $(b-j)$th symbol in $R'_h$ equals $0$ (since
$(b-j)$th symbols alternate in $R'_1,\ldots,R'_{2^b}$ and the $(b-j)$th symbol in $R'_{h+1}$ equals
$1$ as a part of the cascade).

In the following example $b=12$, $i=4$, $j=8$, $\diamond$ denotes irrelevant symbols (not necessarily equal),
$x$ and $\overline{x}$ denote the flipped $(b-i)$th symbol, the $(b-j)$th symbol is underlined:
$$
\begin{array}{rl}
R'_h =&{\diamond}{\diamond}{\diamond}\underline{0}{\diamond}{\diamond}{\diamond}x1111,\\
R'_{h+1} =&111\underline{1}{\diamond}{\diamond}{\diamond}\overline{x}0000.
\end{array}
$$

When we transform $R_k = R'_h$ to $R_{k+1}$, we ``add'' $1$ to the bit corresponding to the $(b-i)$th symbol of $R_k$
and the zero at position $b-j$ will stop carrying the $1$, so that we necessarily have zero among the symbols
$b - i, b - i - 1, \ldots, b - j$ of $R_{k+1}$ (in fact, one can show that they all are zeros except for $b - j$).
Thus, the next ``addition'' of $1$ to the $(b-i)$th symbol of $R_{k+1}$ cannot carry farther than the $(b - j)$th symbol and
so the symbols $b, b - 1, \ldots, b-i+1$ will remain equal to $1$ in $R_{k+1}$ whilst in $R'_{h+1}$ they are all zeros. Therefore, $R'_{h+1} \neq R_{k+1}$.

In the case $k = 2^b, R_k = 11\cdots 100\cdots 0$, with $b - i$ ones, is followed by $R_{k+1} = 00\cdots 0$, with $b$ zeros. But, since $R'_h=R_k$ and $i < j$, we have $R'_{h+1} = 00\cdots 011\cdots 100\cdots 0$, with $j - i$ ones,
after ``adding'' $1$ to the $(b-j)$th symbol of $R'_h$. The case $h=2^b$ is analogous.
%\qed
\end{proof}

%\section{Experimental Results}\label{sec:impl}
\section{Implementation}\label{sec:impl}

To build \RLZPRE{}, we first compute $LZ(T[1,\ell])$ and then
$RLZ(T[\ell+1,n],T[1,\ell])$. For both of them, we utilize the suffix array
of $T[1,\ell]$, which is constructed using the algorithm {\tt libdivsufsort}
\cite{YM,Gog14Sea}.
To compute $LZ(T[1,\ell])$, we use the KKP3 algorithm~\cite{KKP13}.
To compute $RLZ(T[\ell+1,n],T[1, \ell])$, we scan $T[\ell+1,n]$ looking for
the longest match in $T[1,\ell]$ by the standard suffix array based
pattern matching.

The output phrases are encoded as pairs of integers: each pair $(p_j, \ell_j)$
represents the position, $p_j$, of the source for the phrase and the length,
$\ell_j$, of the phrase (note that this is in contrast to the ``distance-length''
pairs $(d_j, \ell_j)$ that we had for encodings).
We then map the output into a sequence of numbers using $2\lceil\log \ell\rceil$-bit
integers with $\lceil\log \ell\rceil$ bits for each pair component.
This is possible because we enforce that our reference size is $\ell \ge \sigma$.

Finally, we compute the LZ parse using a version of KKP3 for large
alphabets, relying on a suffix array construction algorithm for large
alphabets~\cite{QSUFSORT,Gog14Sea}. We then remap the output of LZ to point to positions
in $T$ as described.

\subsection{A recursive variant}
\label{sec:rec}
When the input is too big compared to the available RAM, it is possible that after the first compression step,
\RLZPRE{}, the resulting parse is still too big to fit in memory, and therefore it is still not possible
to compute its LZ parse efficiently.
To overcome this issue in practice, we propose a recursive variant, which takes as input the amount of available
memory. The first step remains the same, but in the second step we make a recursive
call to \OURS{}, ignoring the phrases that were already parsed with LZ, and using the longest
possible $\ell$ value for the given amount of RAM. This recursive process continues until
the LZ parse can be computed in memory.
%This scheme is illustrated in Figure~\ref{fig:recursive_scheme}.
It is also possible to give an additional parameter
that limits the number of recursive calls.
We use the recursive version only in the last set of experiments when
comparing with other LZ parsers in Subsection~\ref{subsec:lz-compare}

\begin{figure}[t]
	\centering
  \includegraphics[width=0.8\textwidth]{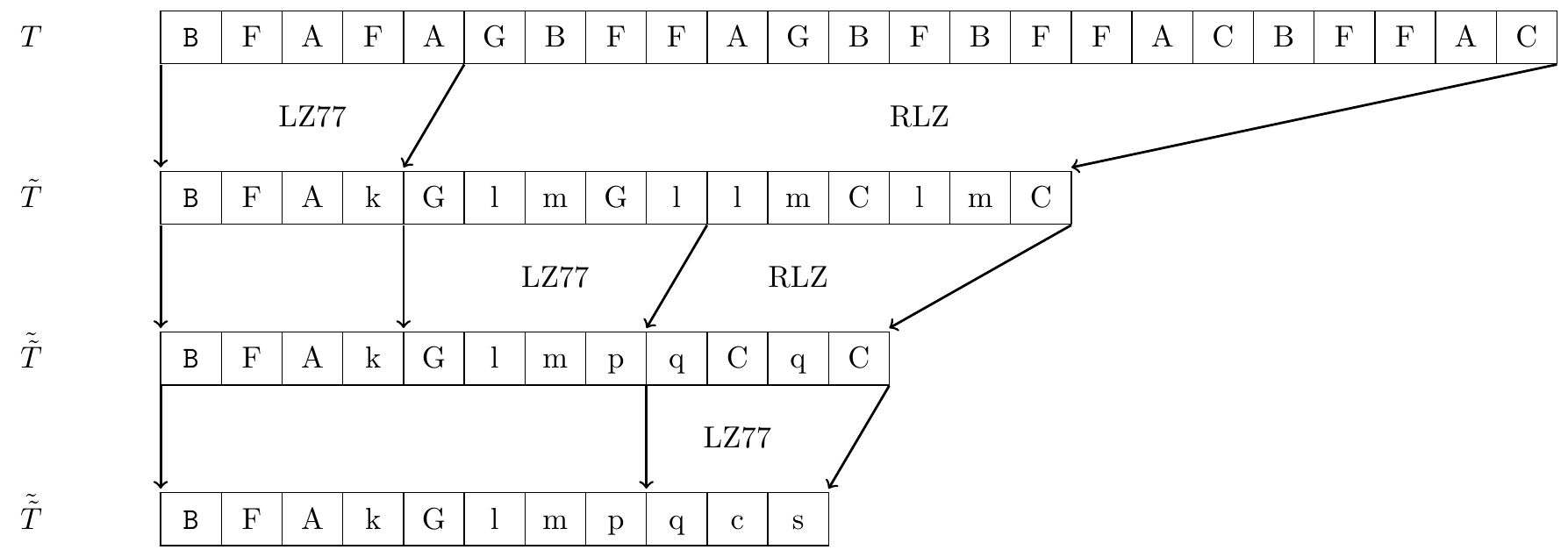}
	\caption{
    Example of the recursive \OURS{} approach, assuming that the available memory limits the computation of LZ to sequences of length $5$.
    The figure only shows the recursive parsing. The rewriting of the phrases proceeds later, bottom up, in a similar fashion as depicted
    in Figure~\ref{fig:example}.
}
	\label{fig:recursive_scheme}
\end{figure}

\subsection{A better mapping}
When the recursive approach is used we need a better mapping from pairs of integers
into integers: the simple approach described above requires $2 \log \ell$ bits for the alphabet after
the first iteration, but in the following iterations the assumption  $\sigma \leq \ell $ may not hold anymore and the amount of bits
required to store the first values may increase at each iteration.
%Computing a perfect hash would resolve the problem, but it could require multiple passes over a
%sequence that does not fit in main memory.
We propose a simple mapping that overcomes this problem.
Let $\sigma_i$ be the size of the alphabet used by the metasymbols after the $i$th iteration.
To encode the metasymbols of the $(i+1)$-iteration we use first a flag bit to indicate whether the phrase is literal or copying.
If the flag is 0, then it is a literal phrase $(c,0)$ and  $\log \sigma_i$ bits are used to store the $c$ value.
If the flag is 1, then it is a copying phrase $(p_i, \ell_i)$ and then  $2 \log \ell $ bits are used to store the numbers.
In this way, after each iteration the number of bits required to store the metasymbols increases only by $1$.

%For the experimental evaluation we used following collections:
%\begin{description}
%    \item[{\tt Wiki}:] Partial dump of english wikipedia pages from 2007.
%    \item[{\tt Kernel}:] Concatenation of different versions of the linux kernel source code.
%    \item[{\tt Cere}:] 50 versions of the original Cere file from pizzachili.
%    \item[{\tt WikiLarge}:] Partial dump of english wikipedia pages from XXXX.
%    \item[{\tt HG}:] $2TB$ of different human genomes, from $1000$ genomes project.
%      Each version has a $0.01\%$ mutation ratio from the previous version.

\begin{table}[ht]
\caption{
  Collections used for the experiments, some basics statistics, and a brief description of their source.
  The first group includes medium-sized collections, from  $45$ to $202$ MiB,
  while the second group consist of large collections, from $22$ to $64$GiB.
  Each group has both regular collections and highly repetitive collections, attested by the
  average phrase length $n/z$.
}
\label{table:data}
\begin{center}
%\resizebox{0.48\textwidth}{!}{
\begin{tabular}{|l|r|r|r|r|r|}
\hline
  Name      &  $\sigma$ &  $n$  &  $n/z$  & Type          & Source \\
\hline
  English   &  225      & 200 MiB  &   15 & English text   & Pizzachili \\
  Sources   &  230      & 202 MiB  &   18 & Source code    & Pizzachili \\
  Influenza &  15       & 148 MiB  &  201 & Genomes        & Pizzachili  \\
  Leaders   &  89       &  45 MiB  &  267 & English text   & Pizzachili  \\
\hline
  Wiki      &  215      & 24 GiB   &  90   & Web pages    &  Wikipedia dumps\\
  Kernel    &  229      & 64 GiB   & 2439  & Source code  &  Linux Kernel   \\
  CereHR    &  5        & 22 GiB   & 3746  & Genomes      &  Pizzachili \\
%\hline
%  HG        &  X        & 2.2 TiB   & XXXX  & Genomes&   $1000$ genomes project \\
%  WikiLarge &  X        & 2.1 TiB   & XXXX  & Web pages &  Wikipedia dumps \\  %% With edit history
\hline
\end{tabular}
%}
 \end{center}
\end{table}

%% WIKI SOURCE: (medium)
%%% https://dumps.wikimedia.org/other/static_html_dumps/April_2007/en/wikipedia-en-html.0.7z

%The experiments involving the largest  collections (between XXX and YYYY $TB$ were run on a machine
%equipped with two six-core 1.9 GHz Intel Xeon E5-2420 CPUs with 15 MiB L3 cache and 120 GiB of DDR3 RAM.

We implemented \OURS{} in C++ and the source code is available under GPLv3 license
in \url{https://gitlab.com/dvalenzu/ReLZ}.
The implementation allows the user to set the value of $\ell$ or, alternatively, to provide the
maximum amount of RAM to be used.
Additionally, scripts to reproduce our experiments are available at  \url{https://gitlab.com/dvalenzu/ReLZ_experiments}.
For the experimental evaluation, we used collections of different sizes and kinds.
They are listed in Table~\ref{table:data} with their main properties.
The experiments were run
on a desktop computer equipped with a  Intel(R) Core(TM) i5-7500 CPU, with $4$ cores, $3.60$GHz  and $16$GB of RAM.

\section{Experimental evaluation}
\label{sec:exp}

\subsection{Entropy coding}
    First we compare the encoded size of \OURS{} with the $k$-order empirical entropy, and also
    with the encoded size of LZ77. For both \OURS{} and LZ77 we used Elias-gamma codes. The results are presented in
    table~\ref{table:entropy}.

    We observe that in the small and low-repetition collections (English and source) \OURS{} requires some extra
    space than $H_k$ for higher values of $k$. This can be attributed to the $o()$ term in our analysis. Also we observe
    the same behavior for LZ77. As expected, for the highly repetitive collections, both $\OURS{}$ and LZ77 use less space
    than the entropy. This is due to the known fact that for highly repetitive collections, $z$ is a better measurement
    of the compressibility than the empirical entropy.
    Therefore, in the following sections, we proceed to study empirically how does \OURS{} compare to LZ77 in terms of
    number of phrases produced by the parsers.

\begin{table}[ht]
\caption{
  Empirical entropy of order $k$ of our collections for $k = 1,2,\ldots6$; and encoded size of ReLZ and LZ77.
  All values are expressed as bits per character (bpc).
}
\label{table:entropy}
\begin{center}
\resizebox{1\textwidth}{!}{
\begin{tabular}{|l|r|r|r|r|r|r|r||r|r||r||}
\cline{2-11}
  \multicolumn{1}{c|}{~} &
  \multicolumn{7}{c||}{Entropy $H_k$} &
  \multicolumn{2}{c||}{ReLZ(bpc); $\ell=$}&
  \multicolumn{1}{c||}{LZ(bpc)} \\
\cline{1-10}
  Name      & $H_0$ & $H_1$ & $H_2$ & $H_3$ & $H_4$ & $H_5$ & $H_6$ & $10$MB & $50$MB &  ~    \\
\hline
  English   &  4.52 &  3.62 &  2.94 &  2.42 &  2.06 &  1.83 &  1.67 & 3.53     & 3.02 & 2.65 \\
  Sources   &  5.46 &  4.07 &  3.10 &  2.33 &  1.85 &  1.51 &  1.24 & 2.97     & 2.64 & 1.94 \\
  Influenza &  1.97 &  1.93 &  1.92 &  1.92 &  1.91 &  1.87 &  1.76 & 0.29     & 0.24 & 0.20 \\
  Leaders   &  3.47 &  1.95 &  1.38 &  0.93 &  0.60 &  0.40 &  0.32 & 0.15     & 0.13 & 0.13 \\
\hline
  Wiki      &  5.27 & 3.86  & 2.35  & 1.49  & 1.08  & 0.86  & 0.71  & 0.79     & 0.80 & 0.56 \\
  Kernel    &  5.58 & 4.14  & 3.16  & 2.39  & 1.92  & 1.58  & 1.32  & 0.02     & 0.02 & 0.018 \\
  CereHR    &  2.19 & 1.81  & 1.81  & 1.80  & 1.80  & 1.80  & 1.80  & 0.02     & 0.02 & 0.013 \\
%\hline
%  HG        &  X        & 2.2 TiB   & XXXX  & Genomes&   $1000$ genomes project \\
%  WikiLarge &  X        & 2.1 TiB   & XXXX  & Web pages &  Wikipedia dumps \\  %% With edit history
\hline
\end{tabular}
}
 \end{center}
\end{table}

\subsection{Effect of Reference Sizes} \label{sec:reference}

We first study how the size of the prefix used as a reference influences the number of phrases
produced by \RLZPRE{} and \OURS{}. These experiments are carried out only using the medium-sized collections,
so that we can run \OURS{} using arbitrarily large prefixes as references
and without recursions.
%, and also without the need to use the recursive version (see Section~\ref{sec:rec}).
We ran both algorithms using different values of $\ell = n/10, 2n/10, \ldots, n$.

The results are presented in Figure~\ref{fig:phrases}.
By design, both algorithms behave as  LZ when $\ell = n$.
\RLZPRE{} starts far off from LZ and its convergence is not
smooth but ``stepped''.
The reason is that at some point, by increasing $\ell$, the reference captures a new sequence that has many repetitions that were
not well compressed for smaller values of $\ell$. Thus \RLZPRE{} is very
dependent on the choice of the reference.
\OURS, in contrast, is more robust since the second pass of LZ does capture much
of those global repetitions.
This results in \OURS{} being very close to LZ even for $\ell = n/10$,  particularly
in the highly repetitive collections.
%TODO
%\textcolor{red}{seria bueno mostrar cual es el costo de memoria de esos valores
%de $\ell$, que tampoco debe ser monotono. Tal vez una fig de approx ratio vs \% of
%space}

\begin{figure}[t]
	\centering
	\includegraphics[width=0.49\textwidth]{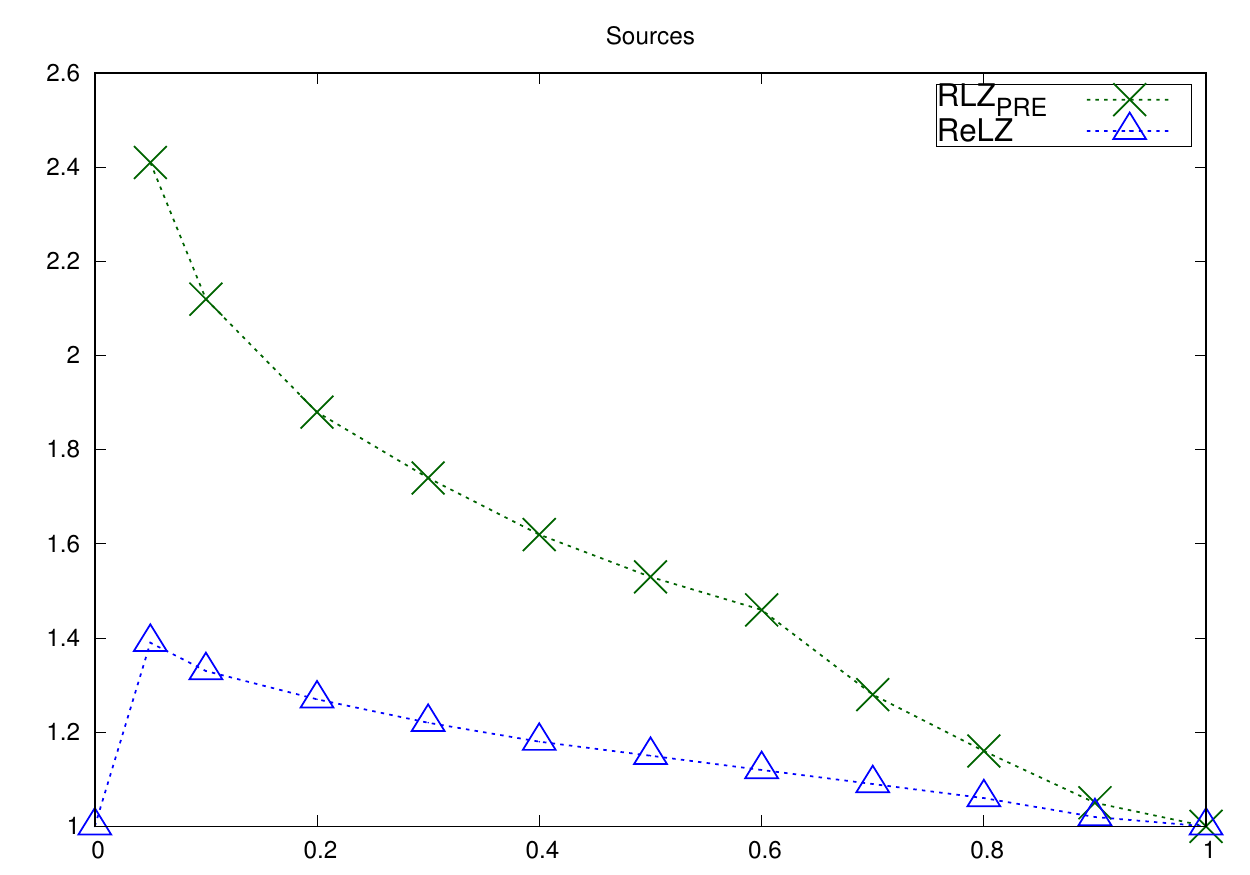}
	\includegraphics[width=0.49\textwidth]{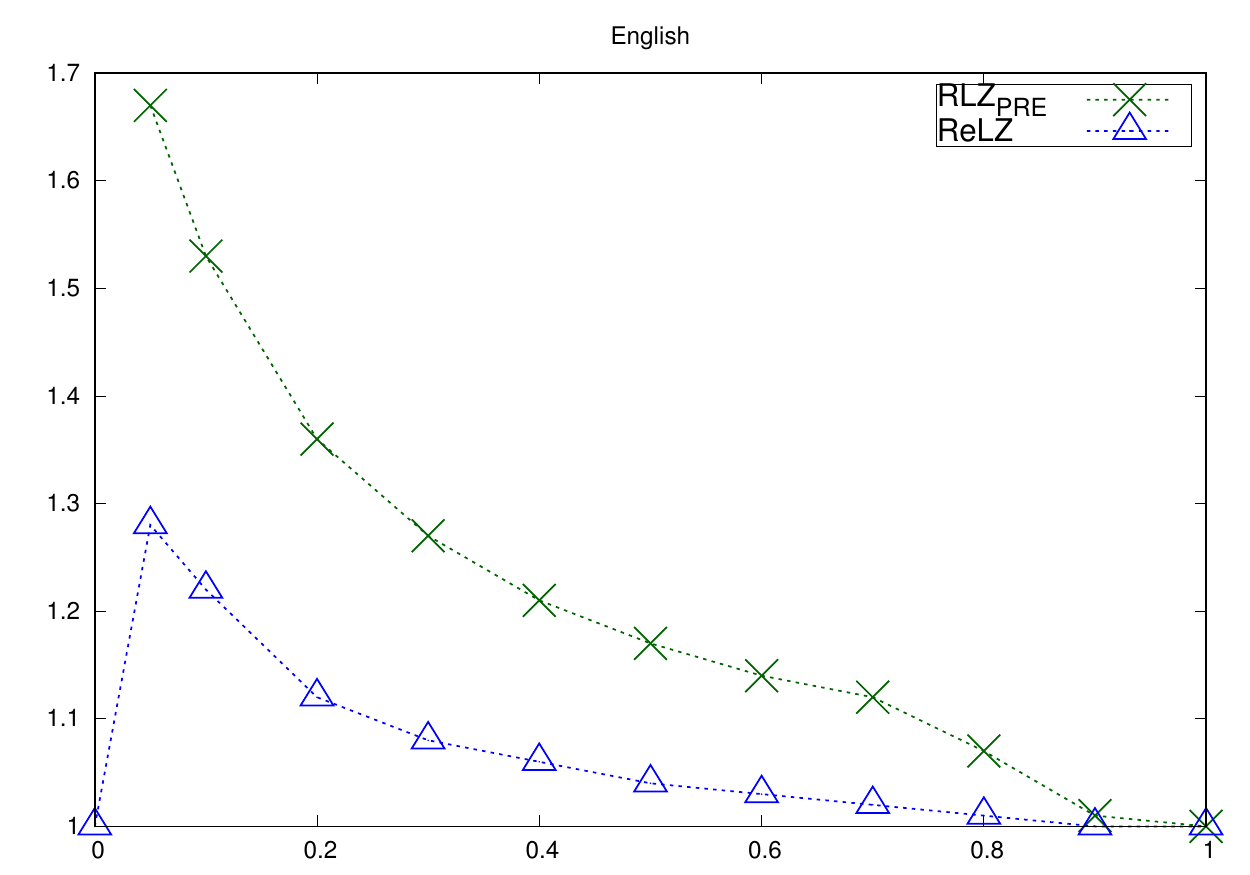}
	\includegraphics[width=0.49\textwidth]{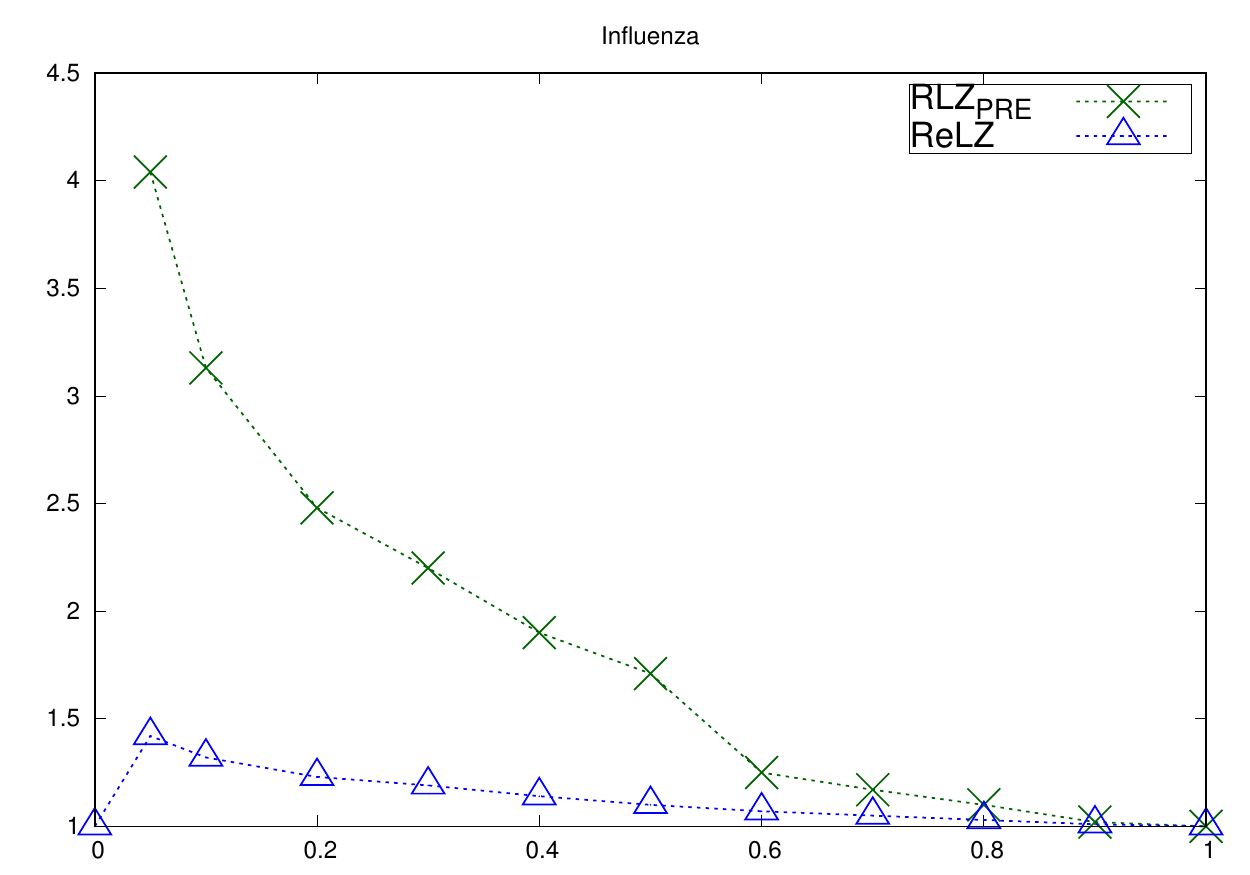}
	\includegraphics[width=0.49\textwidth]{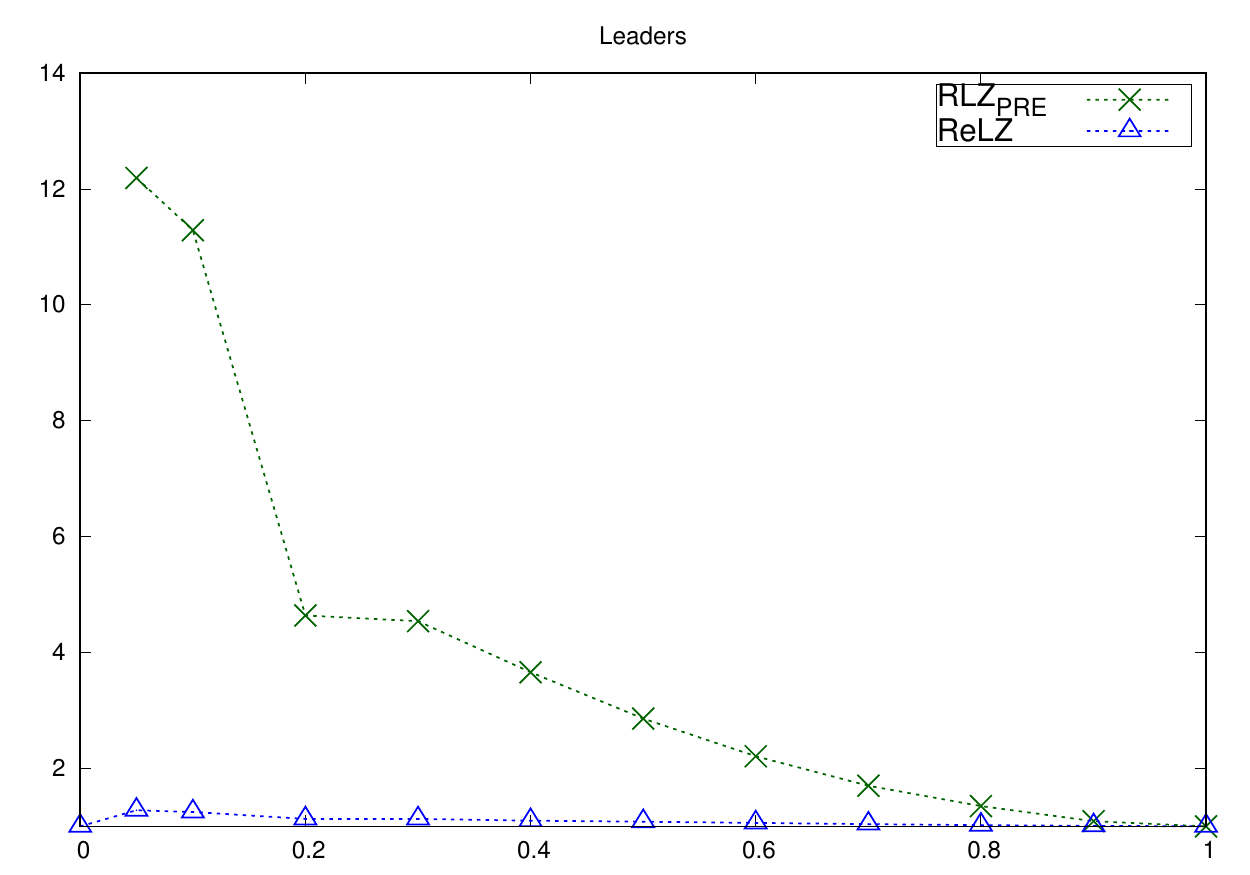}
	\caption{
    Performance of \RLZPRE{} (green) and \OURS{} (blue) for different prefix-reference sizes on medium-sized inputs.
    The y-axis shows the approximation ratio $\hat{z}/z$.
    The x-axis shows $\ell/n$, the size of the prefix-reference expressed as a
    fraction of the input size.
}
	\label{fig:phrases}
\end{figure}

%\begin{figure}[htb]
%	\centering
%	\includegraphics[width=0.49\textwidth]{new_results/sources/Mem}
%	\includegraphics[width=0.49\textwidth]{new_results/english/Mem}
%	\includegraphics[width=0.49\textwidth]{new_results/influenza/Mem}
%	\includegraphics[width=0.49\textwidth]{new_results/einstein.en.txt/Mem}
%	\caption{
%		\OURS{} resource usage for different prefix-reference sizes.
%	\RLZPRE{} (green) corresponds to the first part of
%	our algorithm. The second part of our method is LZ (blue).
%	The others parts consumed negligible resources.
%	Left: peak observed memory.
%	Right: total time to compress.
%}
%	\label{fig:resources}
%\end{figure}

\subsection{Reference Construction}\label{subsec:ref-constr}
As discussed in Section~\ref{sec:algorithm}, the idea of a second compression stage
applied to the phrases can be applied not only
when the reference is a prefix, but also when an external reference is used.
This allows us to study variants of \OURS{} combined with different strategies to build the
reference that aim for a better compression in the first stage.

In this section we experimentally compare the following approaches:
\begin{description}
  \item {\tt PREFIX}: Original version using a prefix as a reference.
  \item {\tt RANDOM}: An external reference is built as a concatenation of random samples of the collection~\cite{HPZ11,GPV16}.
  \item {\tt PRUNE}: A recent method~\cite{LPMW16} that takes random samples of the collections and performs some pruning of redundant parts to construct a better reference.
\end{description}

An important caveat is that methods using an external reference also need to account
for the reference size in the compressed representation because the reference is
needed to recover the output.
%In this section we compare the number of phrases
%of the different methods, so to account for the external reference, we consider the
%number of phrases of its LZ parsing.
For each construction method, we measure the number of phrases produced for
the string ``reference + text'' (only ``text'' for the method {\tt PREFIX})
by the first stage (\RLZPRE{} with prefix equal to the reference)
and by the second stage (LZ on metasymbols corresponding to the phrases),
using three reference sizes: $8$MB, $400$MB, and $1$GB.
We compare the numbers to $z$, the number of phrases in the LZ parsing of the plain text.
This experiment was performed on the large collections and
the results are presented in Figure~\ref{fig:reference_construction}.

\begin{figure}[t]
	\centering
\hspace*{-3mm}
\includegraphics[width=0.5\textwidth]{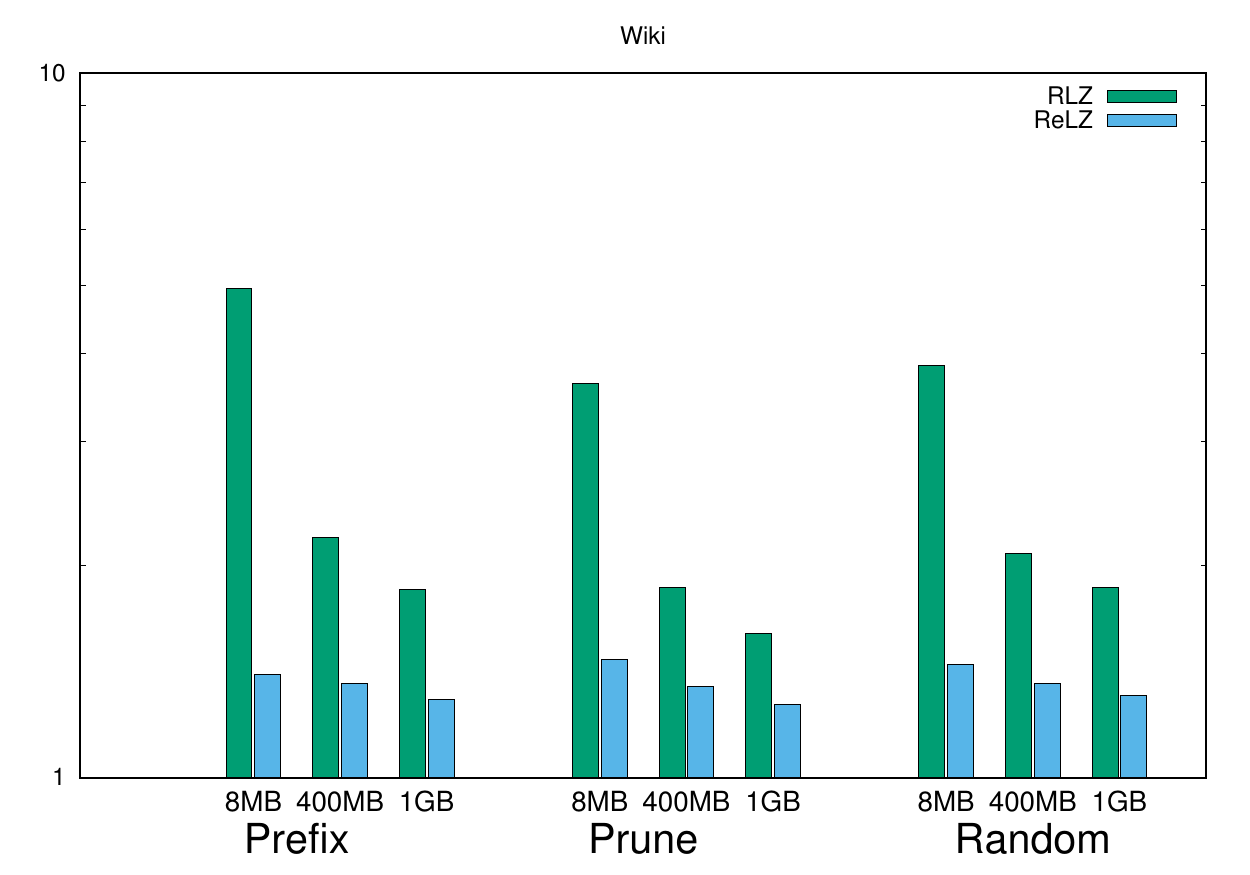}
\hspace*{-2mm}
\includegraphics[width=0.5\textwidth]{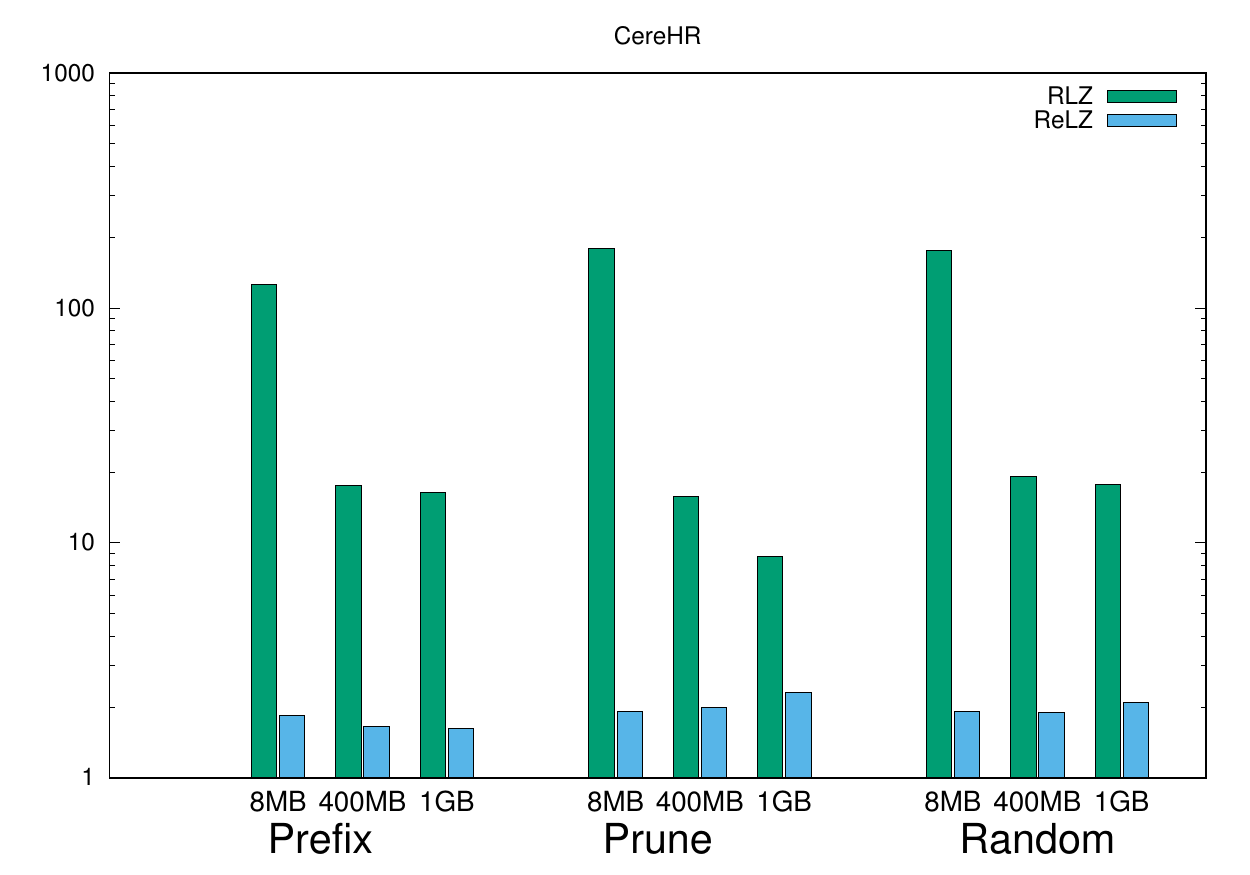}
\hspace*{-2mm}
\includegraphics[width=0.5\textwidth]{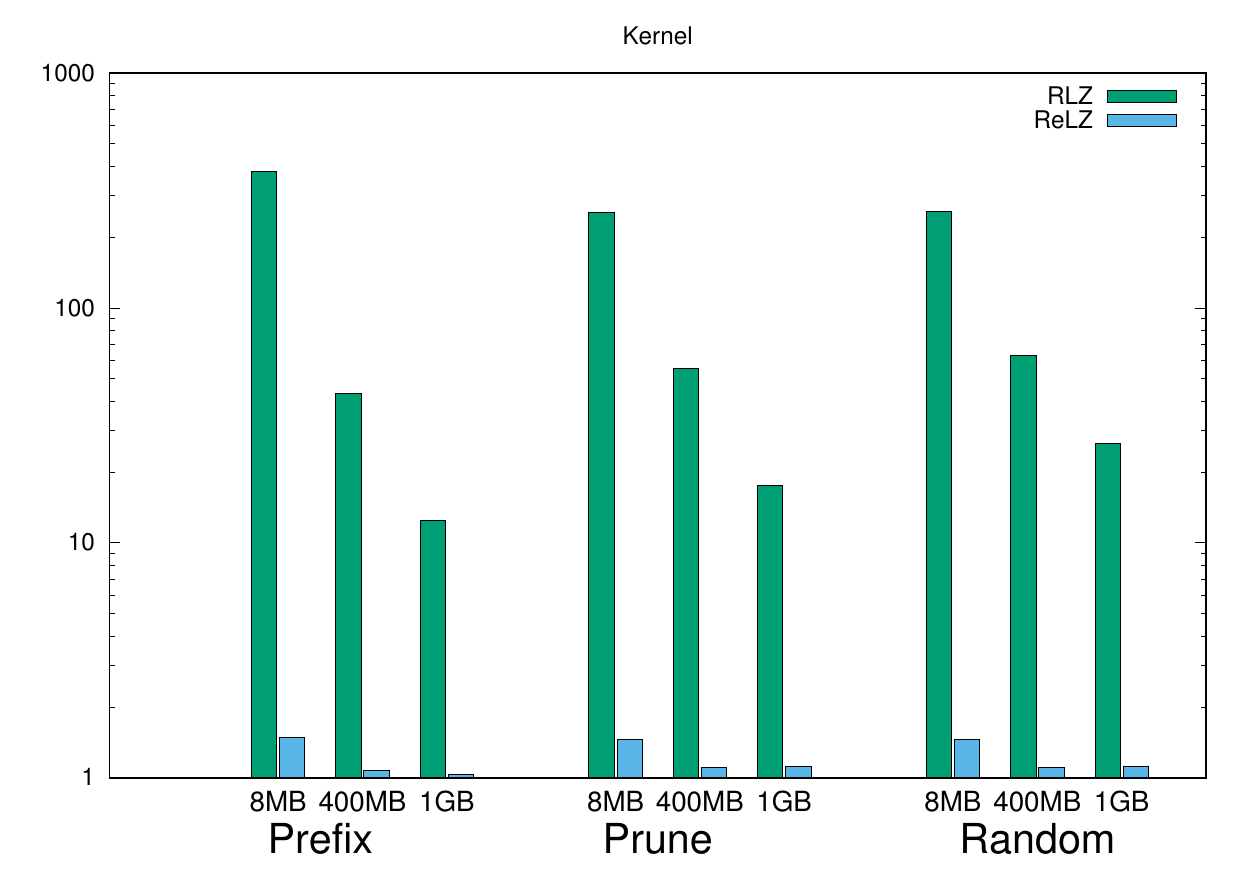}
\hspace*{-3mm}
	\caption{
  Approximation ratio $\hat{z}/z$ for different methods to construct the reference and different reference lengths:
  in green the results after \RLZPRE{}, and in blue after \OURS{}.
  Note that the highly repetitive collections (CereHR and Kernel) use logarithmic scale.
  }
	\label{fig:reference_construction}
\end{figure}

We observe that the second stage of \OURS{} reduces the number of phrases dramatically,
regardless of the reference construction method.
\OURS{} with the original method {\tt PREFIX} achieves the best ratios
as it does not need to account for the external reference.
Depending on the reference size, the approximation ratio in Wiki ranges between $1.4$ and $1.29$,
in CereHR between $1.84$ and $1.63$, and in Kernel between $1.49$ and $1.03$.

Additionally, we observe that although {\tt PRUNE} can improve the results of the \RLZPRE{} stage,
after the second stage the improvements do not compensate for the need to keep an external reference.
This is particularly clear for the largest reference in our experiments.

\subsection{Lempel--Ziv Parsers}\label{subsec:lz-compare}

In this section we compare the performance and scalability of \OURS{} against other
Lempel--Ziv parsers that can also run in small memory (this time, using the recursive version of \OURS{}).

\begin{description}
  \item {\tt EMLZ} \cite{EMLZscan}: External-memory version of the exact LZ algorithm, with memory usage limit set to $4$GB.
  \item {\tt LZ-End} \cite{KK17}: An LZ-like parsing that gets close to LZ in practice.
  \item {\tt ORLBWT} \cite{BGI18}: Computes the exact LZ parsing via online computation of the RLBWT using
    small memory.
  \item {\tt RLZ$_{PRE}$}: Our \RLZPRE{} algorithm (Section~\ref{sec:ReLZ}), with memory usage limit set to $4$GB.
  \item {\tt \OURS{}}: Our \OURS{} algorithm (Section~\ref{sec:ReLZ}), with memory usage limit set to $4$GB.
\end{description}
%% Discuss: exact or approximate, space, etc.

To see how well the algorithms scale with larger inputs, we took prefixes of different sizes of all the
large collections and ran all the parsers on them. We measured the running time of all of the algorithms
and, for the algorithms that do not compute the exact LZ parsing, we also measured the approximation ratio $\hat{z}/z$.
The results are presented in Figure~\ref{fig:scalability_medium}.

\begin{figure}[t!]
	\centering
\includegraphics[width=0.49\textwidth]{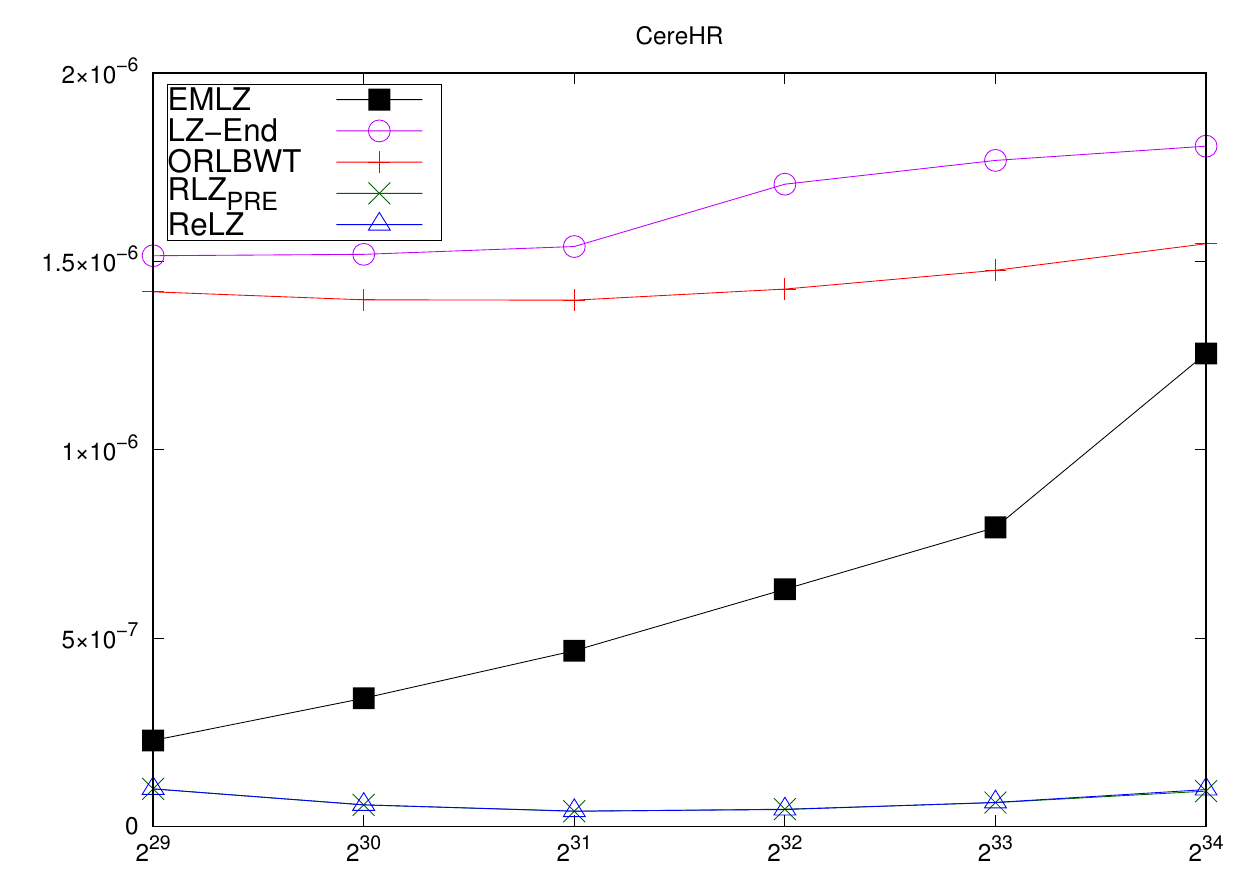}
\includegraphics[width=0.49\textwidth]{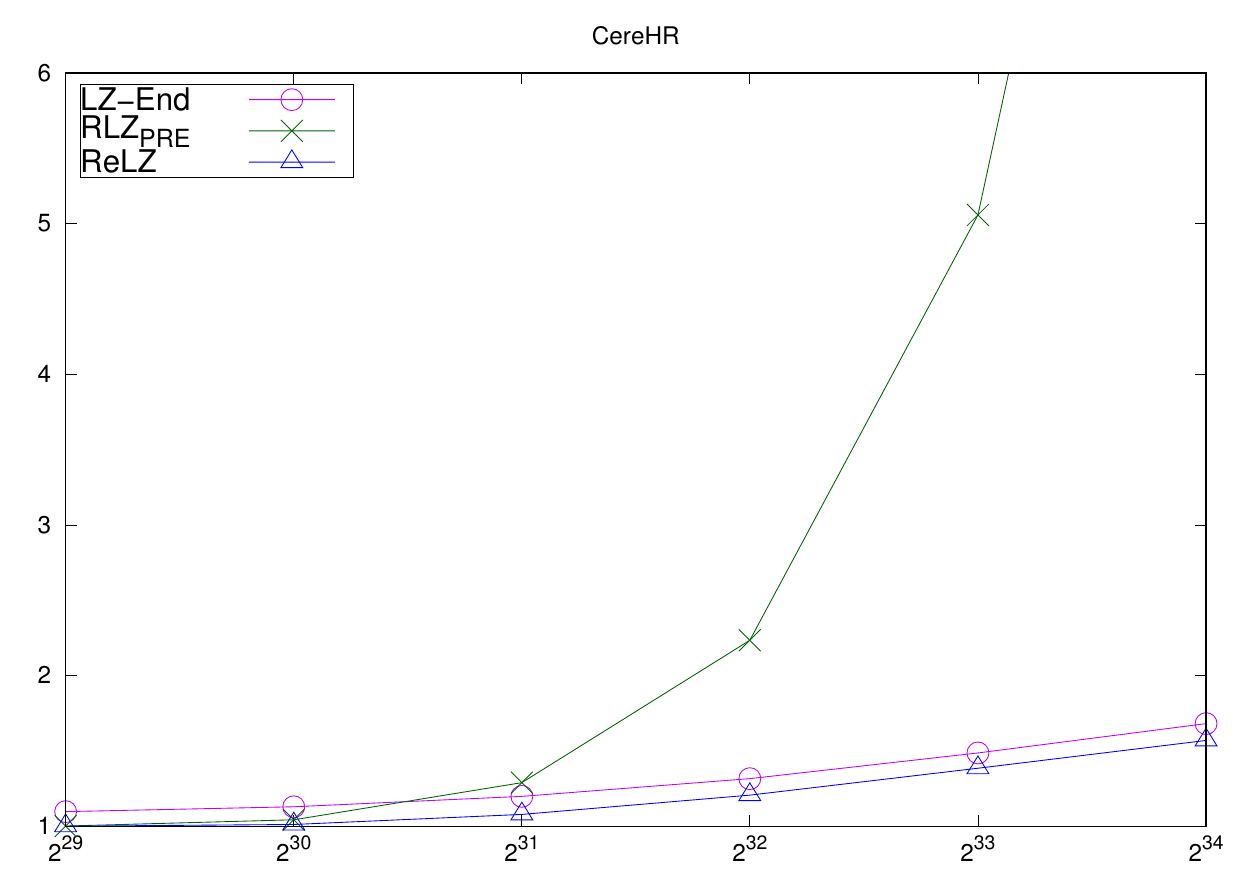}
\includegraphics[width=0.49\textwidth]{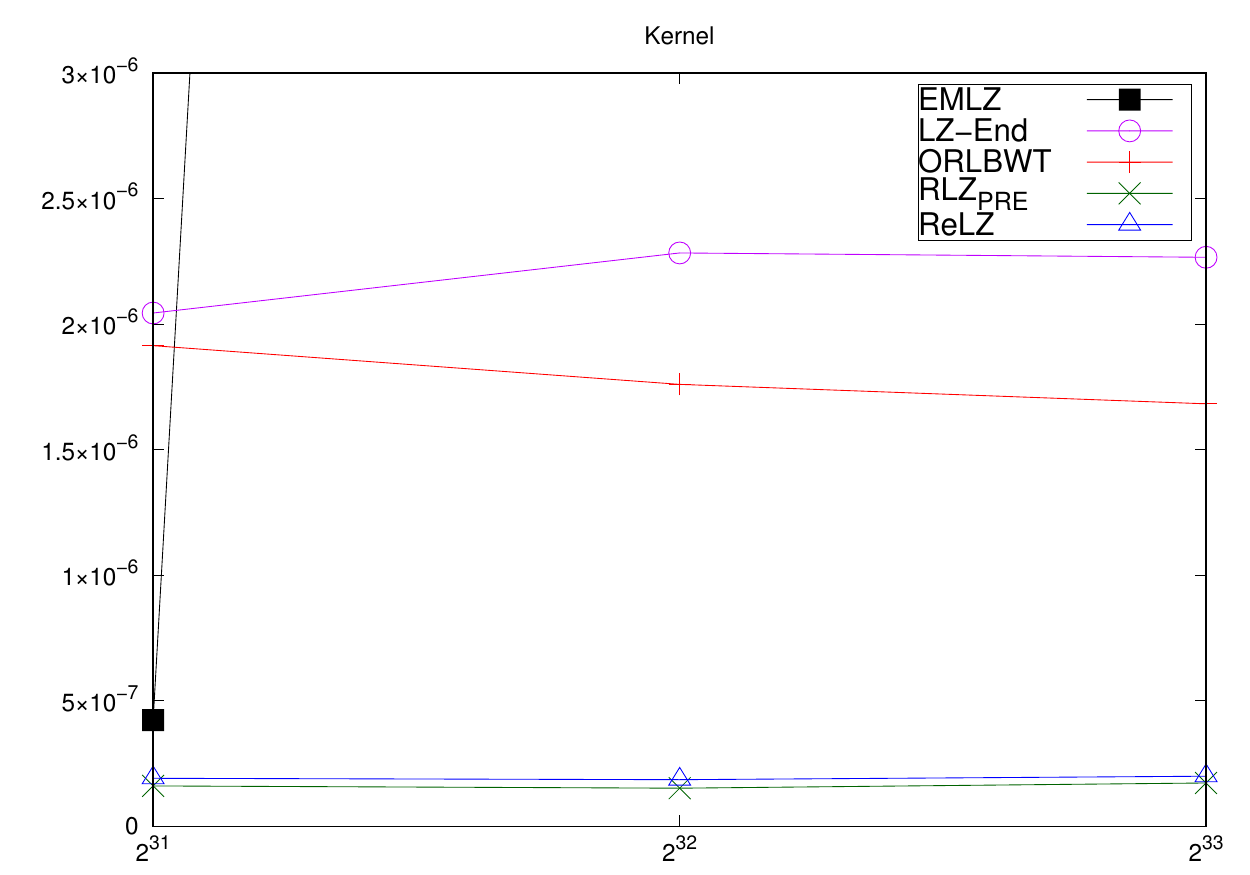}
\includegraphics[width=0.49\textwidth]{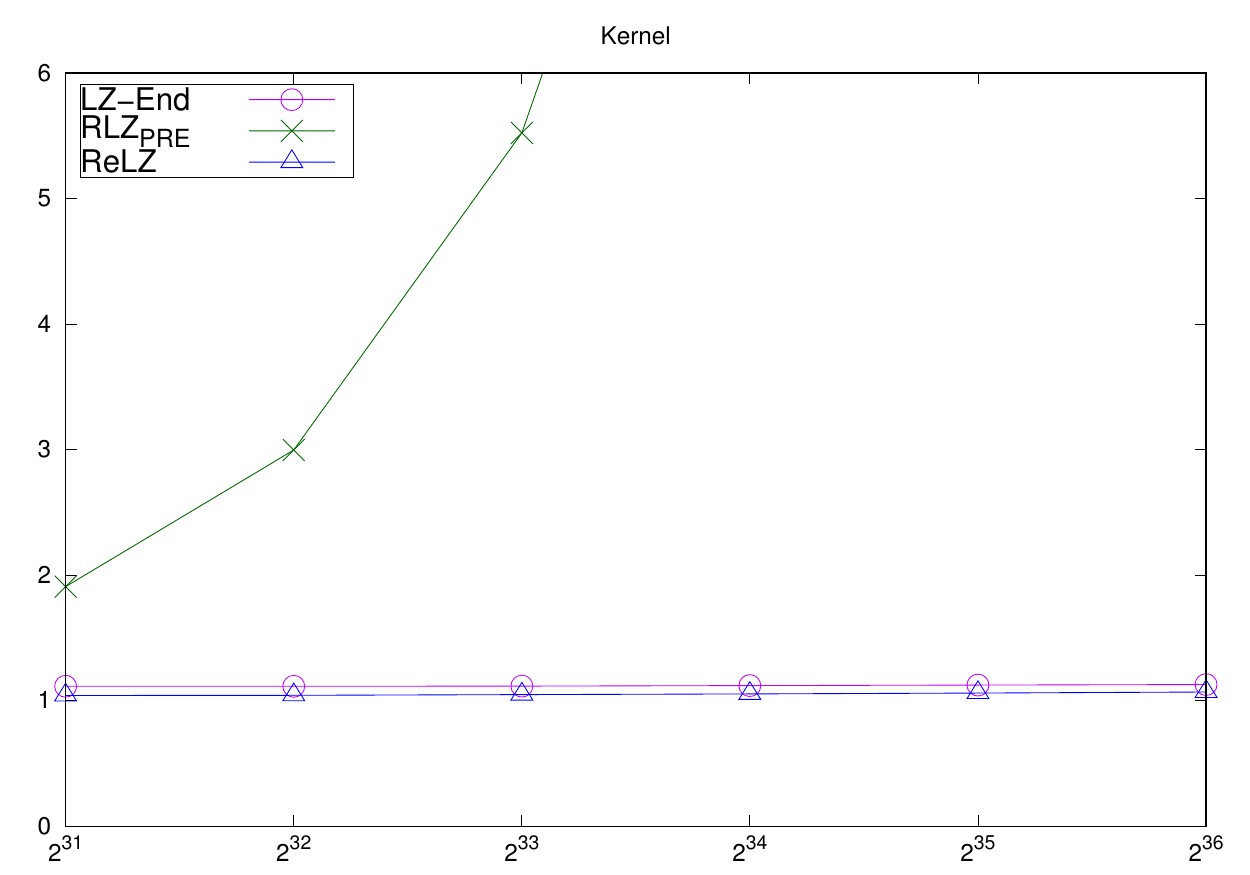}
\includegraphics[width=0.49\textwidth]{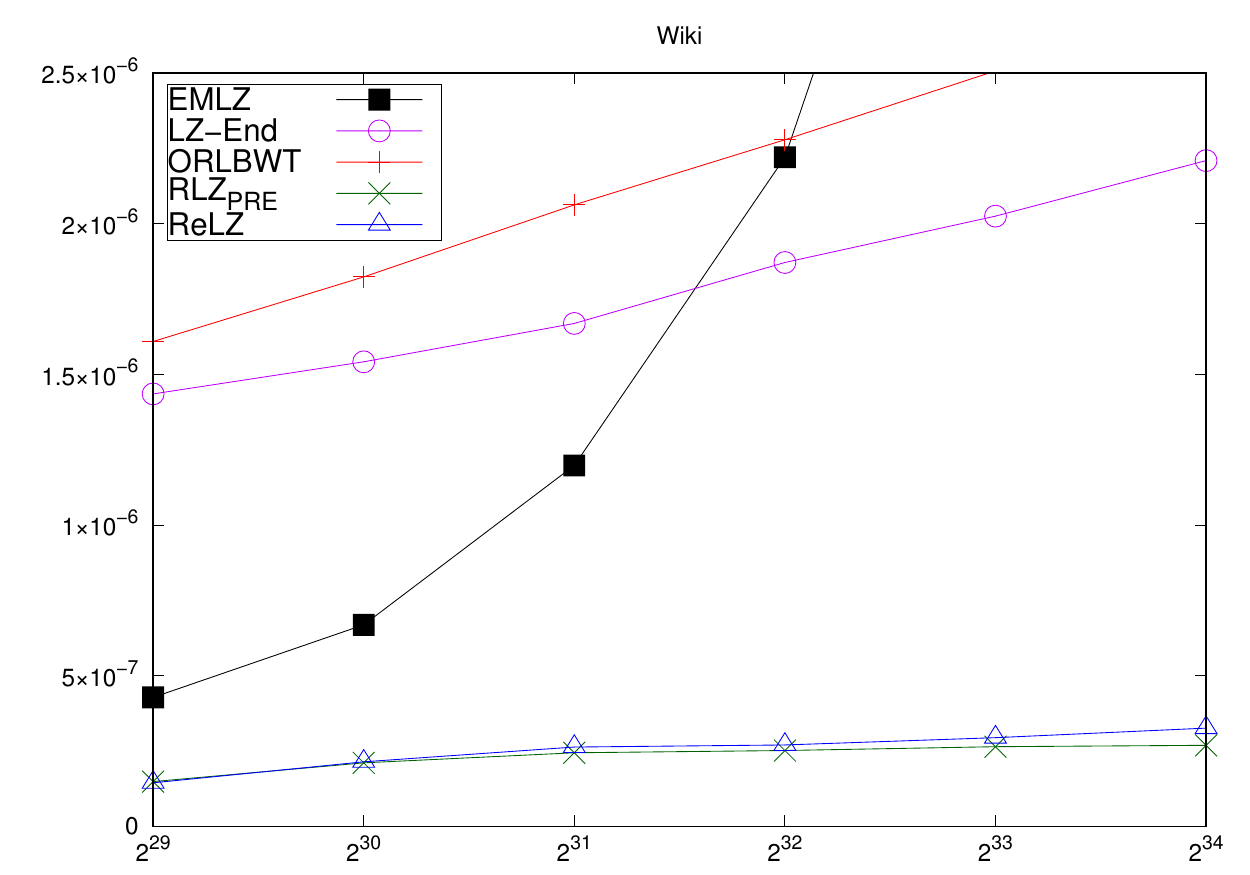}
\includegraphics[width=0.49\textwidth]{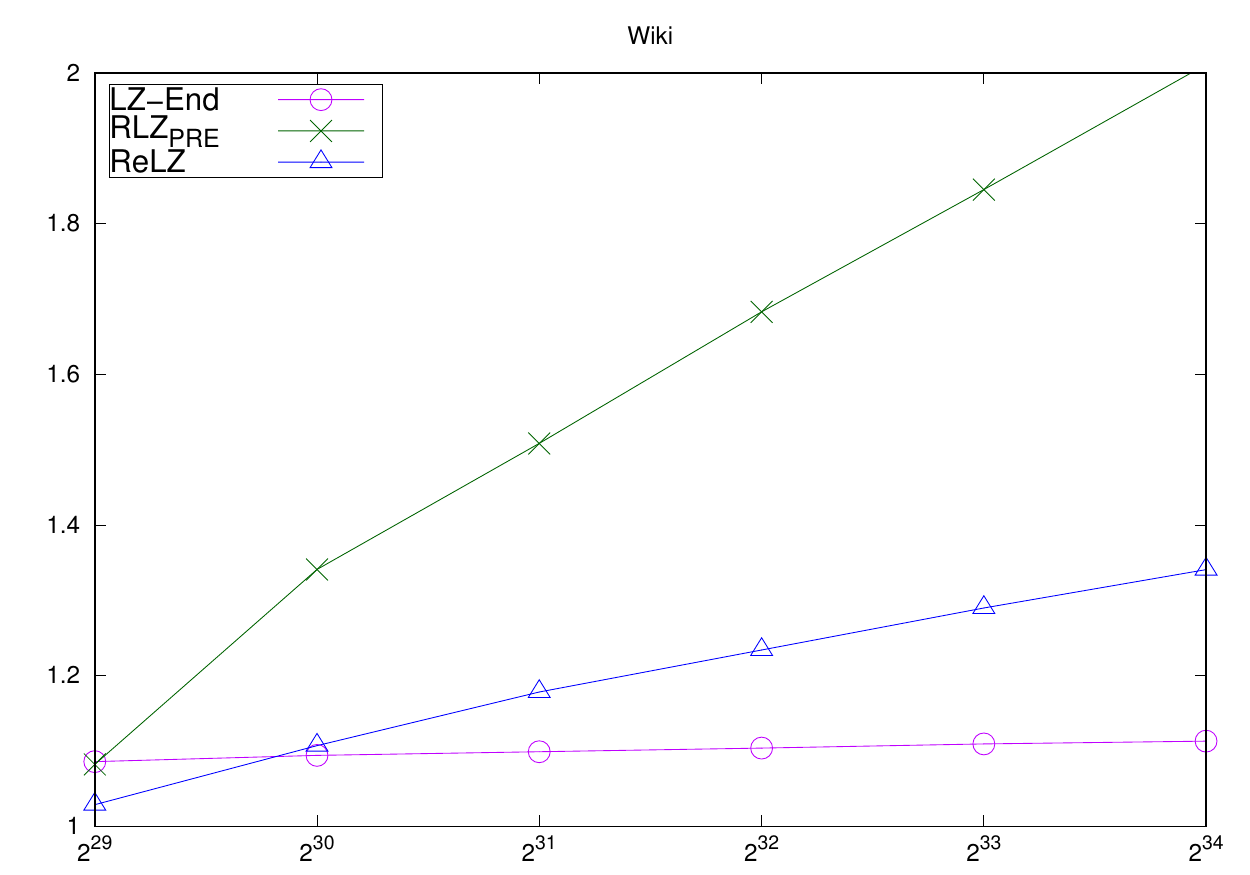}
	\caption{
    Performance of different LZ parsers in the large collections.
    The $x$ axis is the size of the input: increasingly larger prefixes of a given collection.
    Plots on the left show the running time in seconds per MiB.
    Plots on the right show the approximation ratio $\hat{z}/z$.
}
	\label{fig:scalability_medium}
\end{figure}

Figure~\ref{fig:scalability_medium} (left) shows that \OURS{} is much faster than all the previous methods and also that the
speed is almost unaffected when processing larger inputs.
Figure~\ref{fig:scalability_medium} (right) shows that the approximation ratio of \OURS{} is affected very mildly as
the input size grows, especially in the highly repetitive collections. For the normal collections, the approximation
factor is more affected but it still remains below $2$.
%TODO:
%\textcolor{red}{no se puede comparar cuanta memoria usan}

\subsection{Compression ratio}
In this section we study the compression ratio of \OURS{}. We store the $pos$ and $len$ values
in separate files, encoding them using a modern implementation of PFOR codes~\cite{FASTPFOR}
in combination with a fast entropy coder~\cite{FSM}.
We compare against state of the art compressors (LZMA, Brotli) and also agains a very recent
RLZ compressor (RLZ-store). We measure compression ratios, compression times and decompression times
of these tools in the large collections, whose size exceeds the available RAM of the system.

\begin{figure}[htbp]
	\centering
\includegraphics[width=0.49\textwidth]{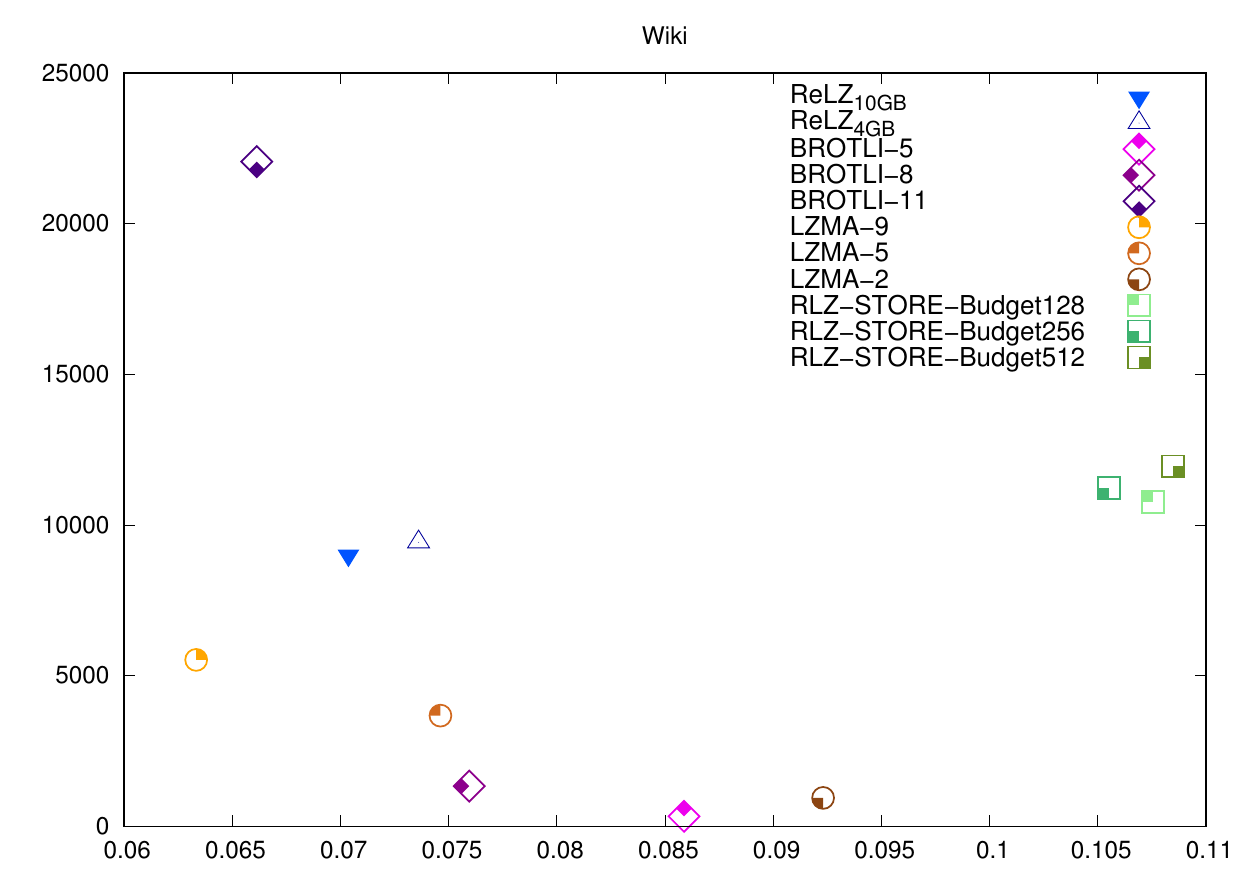}
\includegraphics[width=0.49\textwidth]{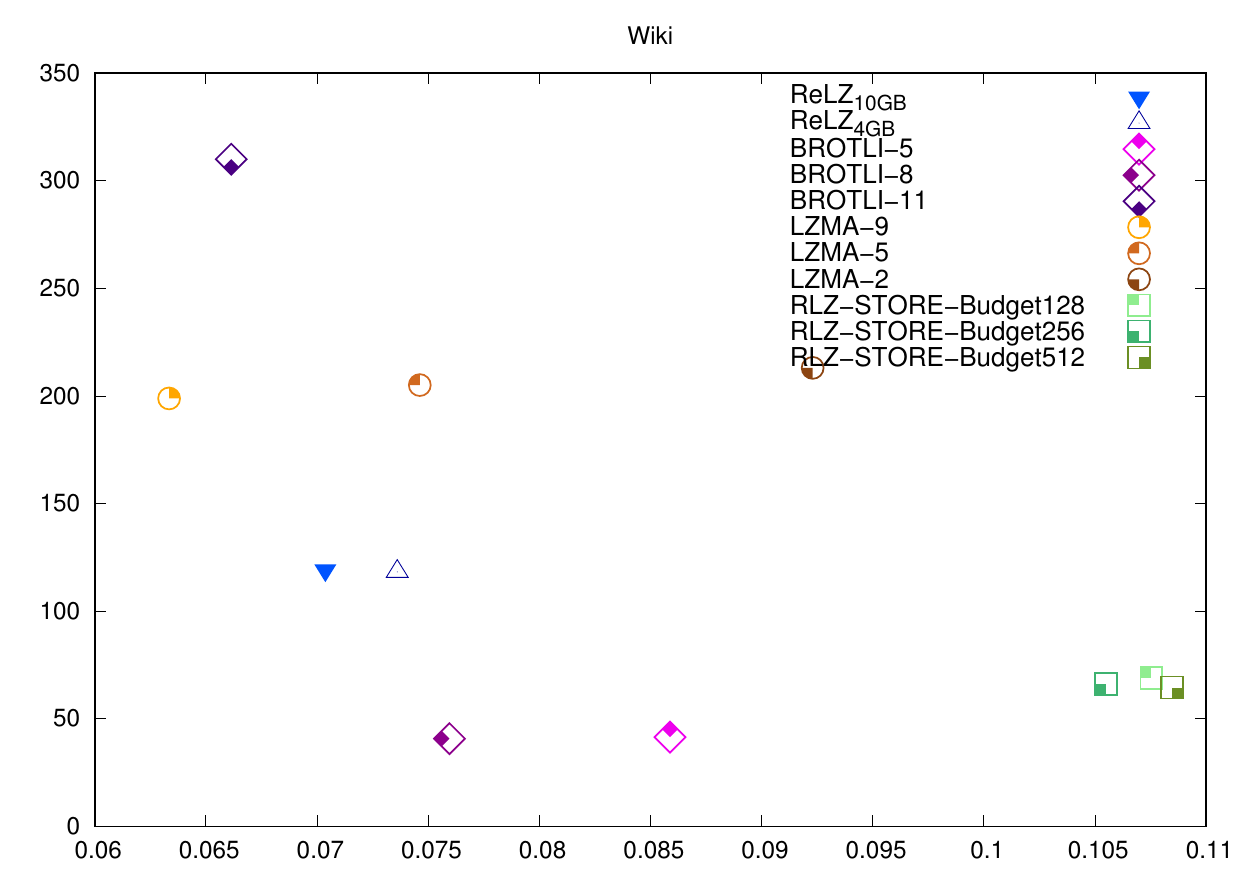}
\includegraphics[width=0.49\textwidth]{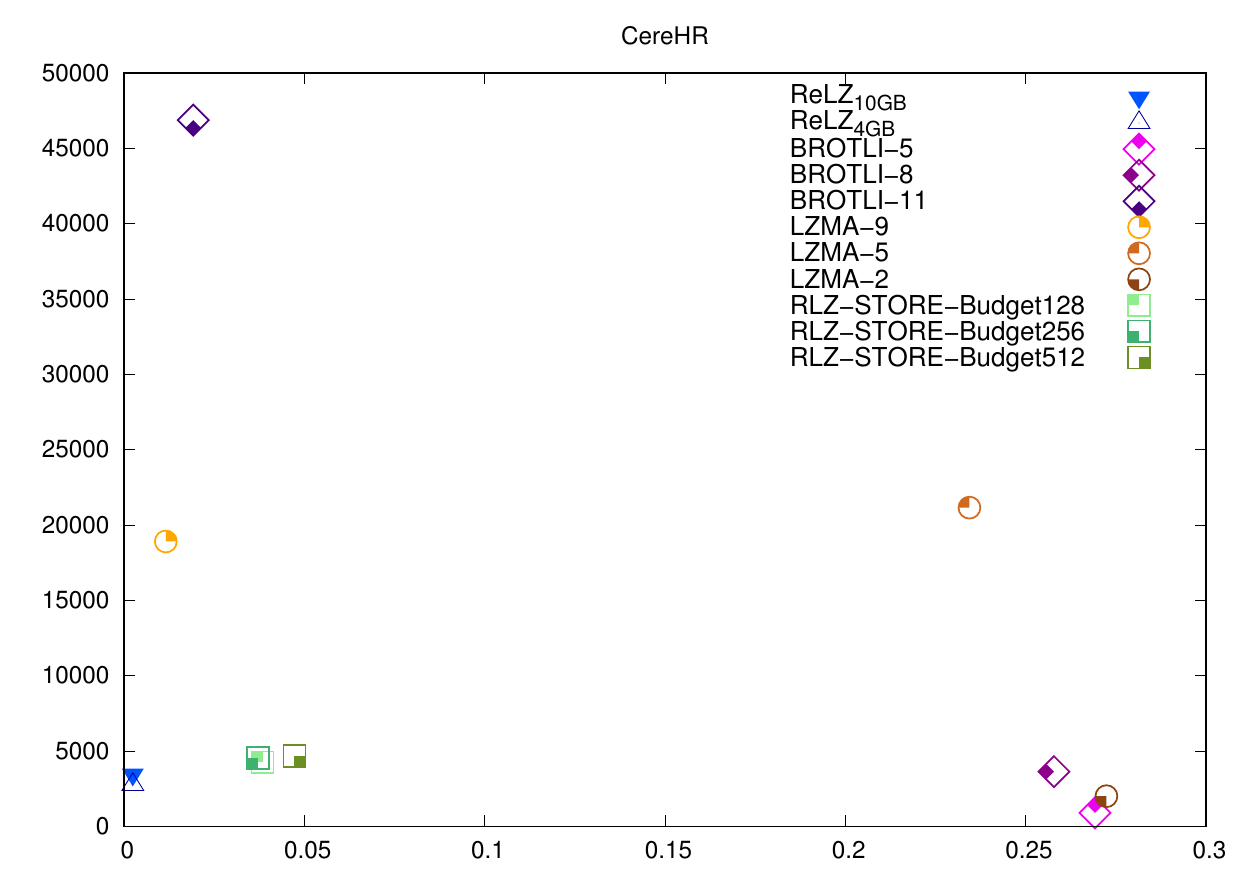}
\includegraphics[width=0.49\textwidth]{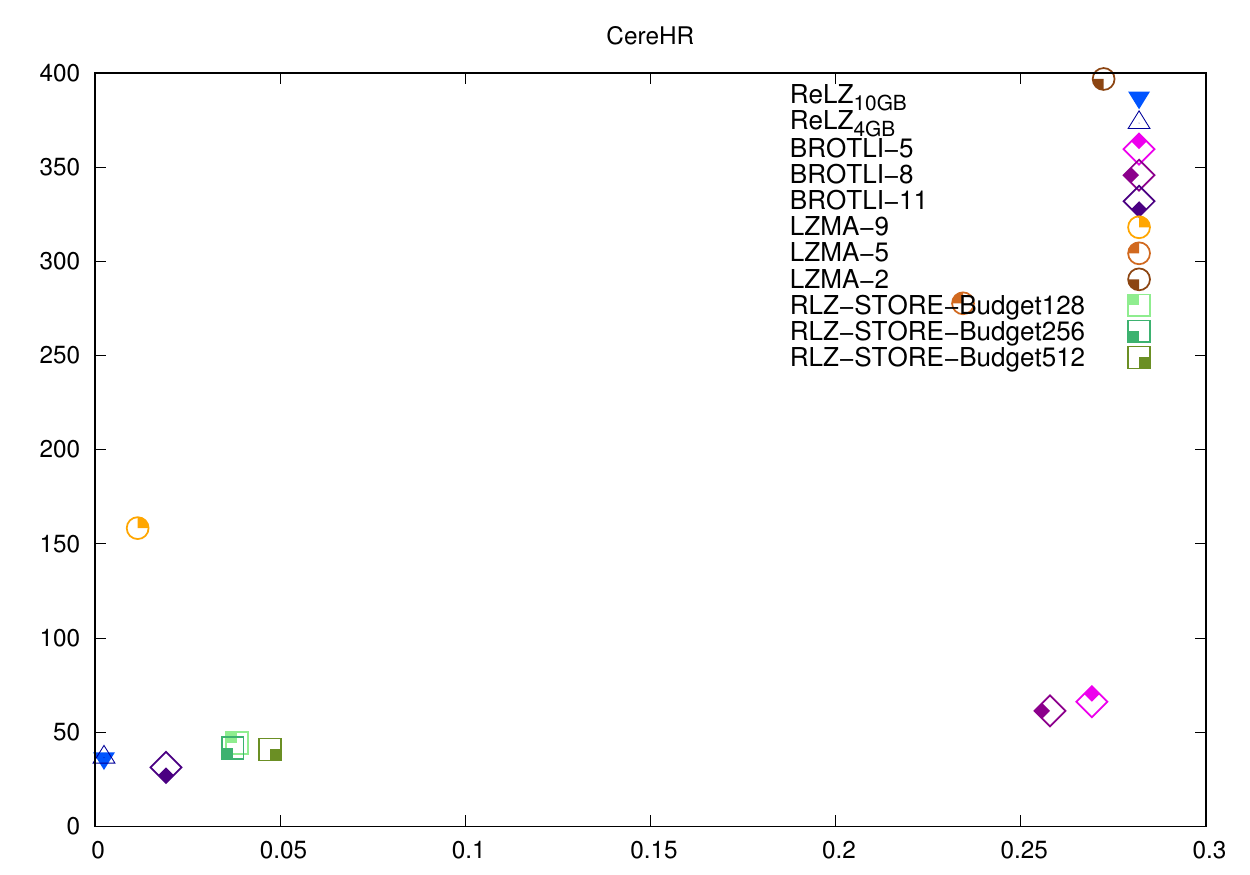}
\includegraphics[width=0.49\textwidth]{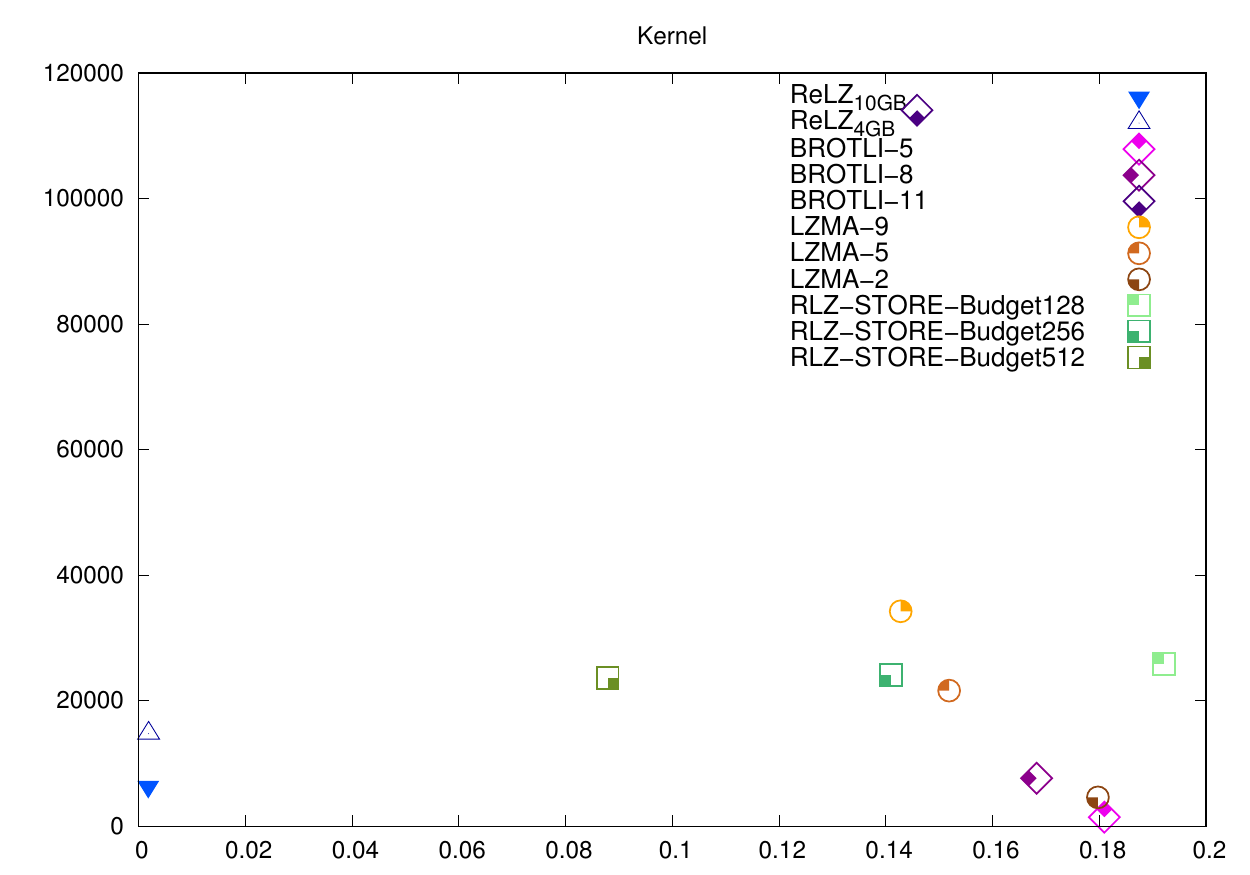}
\includegraphics[width=0.49\textwidth]{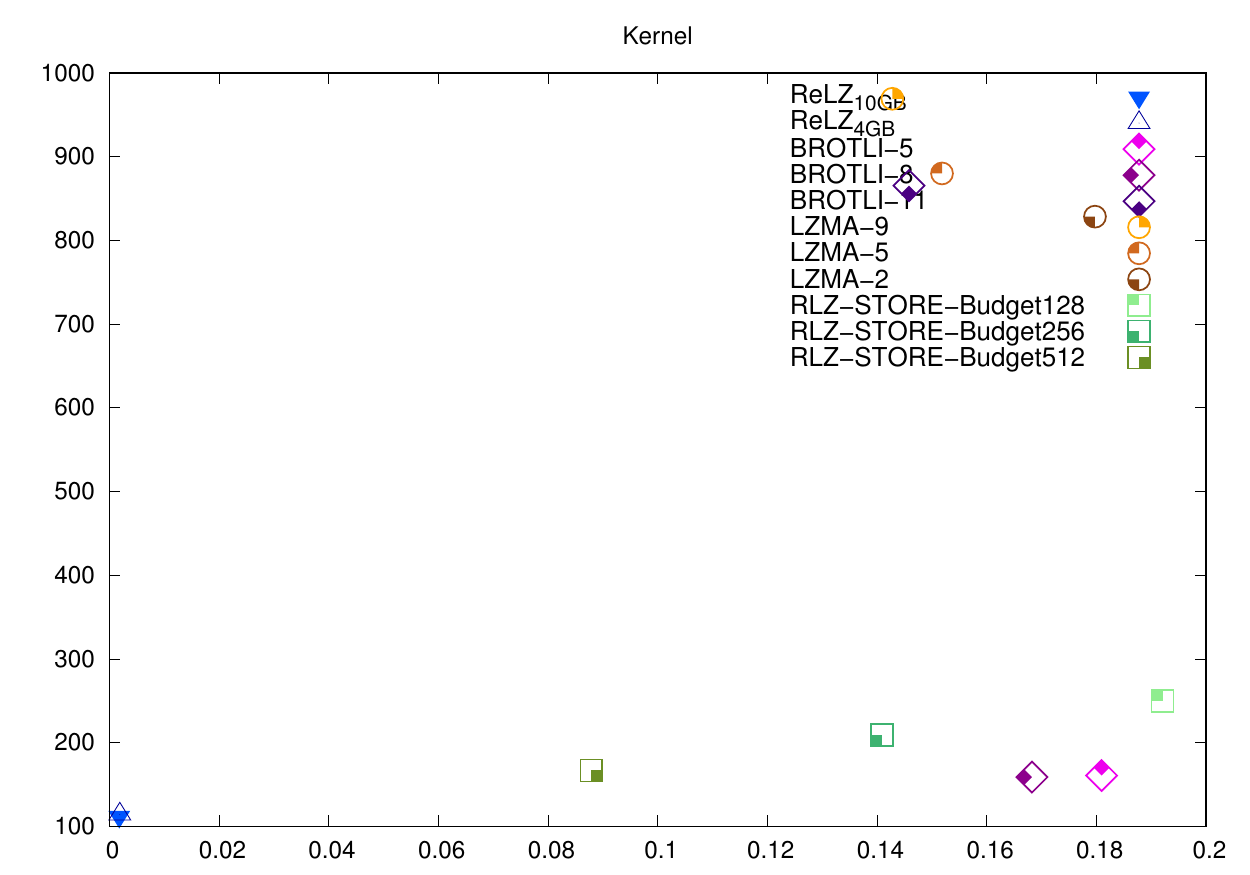}
	\caption{
  Compression results for large collections. The x-axis is the ratio between output and input, and the y-axis is the total compression time in seconds.
  }
	\label{fig:compression_3}
\end{figure}

The results are shown in Figure~\ref{fig:compression_3}.
In the normal collection (Wiki) the performance of \OURS{} is competitive with the state of the art compressors.
In the highly repetitive collections (Cere, Kernel) \OURS{} gives the best compression ratios, with very
similar compression times and competitive decompression times.

Additionally, we run a comparisson again GDC2 and FRESCO.
Both tools are designed to compress a collection of files, using one (or more) as a reference,
and perform referential compression plus second order compression.
GDC2 is specifically designed to compress collections of genomes in FASTA format, and it exploits
known facts about genomes collections (e.g. an important amount of the variations are changes in a single character).
For this, we use a $90$GB collection comprising $2001$ different versions of chromosome $21$.
As expected, GDC2 was the dominant tool, with a compression ratio of $0.00020$, compression time of $15$ minutes and
decompression time of $15$ minutes. \OURS{} compression ratio was $0.00047$, compression time was $49$ minutes and
decompression time was $50$ minutes. We stopped FRESCO execution after $8$ hours, when it had processed slightly more than
half of the collection.

%\section{Conclusions}\label{sec:conclusion}
%
%We have introduced \OURS, a linear-time approximation to the Lempel--Ziv parsing that runs
%within limited memory. \OURS\ combines the Lempel--Ziv parsing of a short prefix of the text,
%a Relative Lempel--Ziv parsing of the remaining text into phrases from the
%prefix, and finally a Lempel--Ziv parsing on the sequence of phrases. Our
%experiments demonstrate that \OURS\ is much faster than other exact and
%approximate algorithms designed to compress very large texts, while retaining
%a very good approximation ratio, usually below $1.5$ and sometimes below  $1.05$.
%The main theoretical challenge is to find an upper bound to the approximation
%ratio of \OURS{} compared to LZ. We conjecture that our $\Omega(\log n)$-factor
%lower bound is indeed tight, and that this is also the upper bound to the
%approximation ratio.

%TODO:
%
%* Discuss future work/open problems
%
%- Prove approximation ratio
%
%- Prove entropy.
%
%- Can we recompress a (pos, len, mismatch char) output ?
%
%- Encodings to achieve smaller size
%
%This process can be done in parallel: we read large blocks of the text into memory
%and we can have $t$ threads each parsing a chunk of the block. Each thread needs only to read the SA
%and its assigned chunk of $T[\ell+1,n]$.
%After all the threads have parsed theirs chunk their output is written sequentially according to their
%chunks.

\subparagraph{Acknowledgements}
This work started during Shonan Meeting 126 ``Computation over Compressed Structured Data''. Funded in part by  EU's Horizon 2020 research and innovation programme under  Marie Sk{\l}odowska-Curie grant agreement No 690941 (project BIRDS).

%\label{sec:discussion}

%\newpage

\bibliographystyle{plain}
\bibliography{biblio}

\end{document}